\documentclass[lettersize,journal]{IEEEtran}
\usepackage{amsmath,amsfonts}
\usepackage{array}
\usepackage[caption=false,font=normalsize,labelfont=sf,textfont=sf]{subfig}
\usepackage{textcomp}
\usepackage{stfloats}
\usepackage{url}
\usepackage{verbatim}
\usepackage{graphicx}
\usepackage{cite}
\usepackage{amssymb}
\usepackage{amsthm}

\graphicspath{{./figure/}}
\newtheorem{problem}{Problem}
\newtheorem{lemma}{Lemma}
\newtheorem{remark}{Remark}

\newtheorem{theorem}{Theorem}

\usepackage{xcolor,soul}
\colorlet{mycolorname}{yellow!20}

\sethlcolor{mycolorname}

\usepackage{algorithm, algpseudocode}

\begin{document}

\title{Contouring Error Bounded Control for Biaxial Systems with Structural Flexibility and Input Delay}

\author{Meng Yuan, ~\IEEEmembership{Member,~IEEE}, Tianyou Chai,~\IEEEmembership{Life Fellow,~IEEE}
\thanks{This work was supported in part by European Union (EU)-funded Marie Sklodowska-Curie Actions (MSCA) Postdoctoral Fellowship under grant number 101110832 and in part by the Research Program of the Liaoning Liaohe Laboratory under number LLL23ZZ-05-01.

Meng Yuan is with the School of Engineering and Computer Science, Victoria University of Wellington, Kelburn 6012, New Zealand (emails: meng.yuan@vuw.ac.nz).

Tianyou Chai is with the State Key Laboratory of Synthetical Automation for Process Industries, Northeastern University, Shenyang 110819, China (email:
tychai@mail.neu.edu.cn).}

}

\maketitle

\begin{abstract}

Precision contouring control is crucial in industrial machining processes, particularly for applications such as laser and water jet cutting, where contouring accuracy directly determines product quality. This paper presents a novel control strategy for biaxial machines featuring position-dependent flexibility and input delays, ensuring that the end-effector accurately traverses the desired contour within specified contouring error bounds and system constraints. To capture the rotation dynamics for systems with mechanical vibration, we introduce a high-fidelity model and explicitly consider the input delay with augmented system states. The controller design is based on the model predictive control scheme to enforce system states staying in robust control invariant sets defined by the reference model and switched linear time-invariant control-oriented models. The proposed algorithm is not restricted to a specific shape of the curve that is being traversed. The effectiveness of the proposed control algorithm is demonstrated in an experimental environment with discretizations and input delay. The results show that a bounded contouring error can be achieved by the proposed method in a performance degradation environment with a low commissioning effort.

\end{abstract}


\section{Introduction}

\IEEEPARstart{C}{ontouring} error is defined as the shortest distance between the actual motion path and the prescribed reference trajectory, representing the degree to which the end-effector deviates from the ideal path during operation \cite{yuan2019bounded, wang2021global}. Such deviation directly affects machining precision, leading to over-cutting, undercutting, or surface quality issues, and can further result in severe assembly and functional failures, ultimately reducing product consistency and production efficiency. Therefore, precise control of contouring error is crucial as it forms the cornerstone for achieving high-precision manufacturing and efficient automation \cite{yuan2019bounded}.

The dual-drive gantry machine, with relatively higher structural rigidity and larger operating envelopes, has seen increasing interest in precision engineering applications \cite{li2018modeling}. Although the intensive usage of finite-element optimisation in the mechanical design process enhances the rigidity of structural parts in machines, the non-synchronized movement of dual drives can still happen due to the flexible linkages in machines. The resulting rotational dynamics lead to the position discrepancy between actuators and end-effector, thus deteriorating the contouring accuracy \cite{sun2024composite}. This highlights the need of considering the flexible structure in the controller design to refine contouring performance.

A range of work has been conducted to improve the contouring performance. The first category covers the cross-coupled control (CCC) based on estimated contour error \cite{koren1980cross, aarnoudse2022cross}. The CCC was proposed to minimise the contouring error by correcting the input command and forcing the axial tool position onto the path. In the original version of CCC, the corrective command is calculated from the weighted contour error and forwarded equally to axis controllers with individual axis controllers unchanged. However, the constant compensator value and cross-coupling gains make the conventional CCC ineffective in dealing with nonlinear contours and even lead to oscillation for linear contours when steady-state error tends to zero \cite{koren1991variable, kuang2019precise}. To deal with these drawbacks, methods such as introducing variable cross-coupling gains were proposed \cite{li2022design}. The cross-coupled control based approach is still prevalent in contouring control and the variants of CCC are related to the choice of contour error model and the control laws \cite{hu2021novel,aarnoudse2024cross,zhao2024cross}.

Another category of contouring control is based on the coordinate transformation where decoupled controllers are designed to minimise the transformed errors \cite{dao2021high, liu2022robust}.  The coordinate transformation-based frames simplify the controller design by decoupling the tangential and contour dynamics and its application can be found in biaxial applications including robotics \cite{chen2010coordinate} and XY tables \cite{yao2011orthogonal, lou2013task, yang2019novel}. However, the difficulty in obtaining the coordinate transformation matrix for free-form contour and the assumption that the dynamics of two axes should be small limits the applicability of the coordinate transformation based methods.

As a controller that explicitly considers operating constraints, model predictive control (MPC) has been widely applied in contouring control scenarios where such constraints are of critical importance \cite{lam2012model, yang2015pre, li2022control}. In \cite{lam2012model}, a model predictive contouring control method was proposed to deal with competing objectives, improving the contouring accuracy and minimising the traverse time. In \cite{yang2015pre}, an MPC framework was introduced to pre-compensate servo contour errors by incorporating predicted contour deviations into the reference trajectory, minimising their squared sum under kinematic constraints. To enhance contour accuracy and tracking performance in a biaxial permanent magnet linear synchronous motor system, \cite{li2022control} presents a strategy that integrates MPC with a model reference adaptive system and variable gain cross-coupling, effectively reducing dynamic mismatches and improving the robustness of system. However, neither \cite{yang2015pre} nor \cite{li2022control} provides a formal guarantee on the contouring error bound, nor do they account for structural flexibility or input delays, which motivates the present work.

Despite numerous works aimed at reducing the contouring error, several critical research gaps remain. From a practical standpoint, industrial processes such as laser and water jet cutting impose strict contouring tolerances, where even small violations can result in defective products. From a theoretical perspective, existing contouring control methods primarily focus on error minimisation without providing formal guarantees that the contouring error remains within a specified bound. First, achieving bounded contouring error while respecting system state and input constraints remains a challenge, particularly for systems with position-dependent flexibility and input delays. Second, input delays commonly encountered in industrial systems complicate controller design and may compromise system stability. Third, in flexible systems lacking real-time end-effector feedback, contouring performance validation typically relies on post-process product quality, making real-time algorithm verification challenging.

To address these challenges, this work focuses on a biaxial system with structural flexibility and input delay. To the best of the authors' knowledge, no existing contouring control method provides a formal guarantee on the contouring error bound for systems with position-dependent structural flexibility and input delay. The main contributions are summarised as follows:

\begin{enumerate}
    \item High-fidelity modelling: A high-fidelity system model is developed to capture position-dependent flexibility and end-effector dynamics, with explicit incorporation of input delays. The model is validated for both accuracy and applicability in systems with limited structural rigidity.
    \item Guaranteed contouring performance: An MPC-based control strategy is developed, in which robust control invariant (RCI) sets are employed to ensure that contouring errors remain bounded while meeting system constraints. The proposed method provides guarantees of stability and feasibility, and consistently satisfies contouring accuracy specifications in real-time operation.
    \item Experimental validation for industrial relevance: To address the absence of real-time end-effector feedback in practical applications, a back-to-back motor-based hardware-in-the-loop (HIL) platform is developed. This platform emulates realistic industrial conditions, including discretisation effects and input delays, thereby enabling rigorous evaluation of the proposed algorithm’s effectiveness under practical digital control scenarios.
\end{enumerate}

\emph{Notation}: $\mathbb{R}$ is the set of real numbers. $\mathbb{Z}_{0+}$ and $\mathbb{Z}_{[m,n]}$ are the set of non-negative integers and integer intervals with $m$ and $n$ included. A $m\times n$ zero matrix is denoted by $0_{m,n}$. Consider $a\in \mathbb{R}^{n_{a}}$, $b \in \mathbb{R}^{n_{b}}$, the stacked vector is represented as $(a,b) \triangleq [a^{T}, b^{T}]^{T} \in \mathbb{R}^{n_{a}+n_{b}}$. The set $\mathbb{B}^{n}(\rho)$ is a closed ball in $\mathbb{R}^{n}$ with $\rho$ radius in infinity norm. The vector $x(k)$ is the value measured at time instant $k$ and vector $x_{(i|k)}$ represents the predicted value of $x$ at time instant $k+i$ based on the data sampled at $k$. For two sets $\mathcal{X}$ and $\mathcal{Y}$, $\mathcal{X} - \mathcal{Y}$ represents the Minkowski difference.

\section{System Description and Problem Formulation} \label{sec:prob_formu}
\subsection{Configuration and modelling}
\begin{figure}[tb]
	\centerline{\includegraphics[width=\columnwidth]{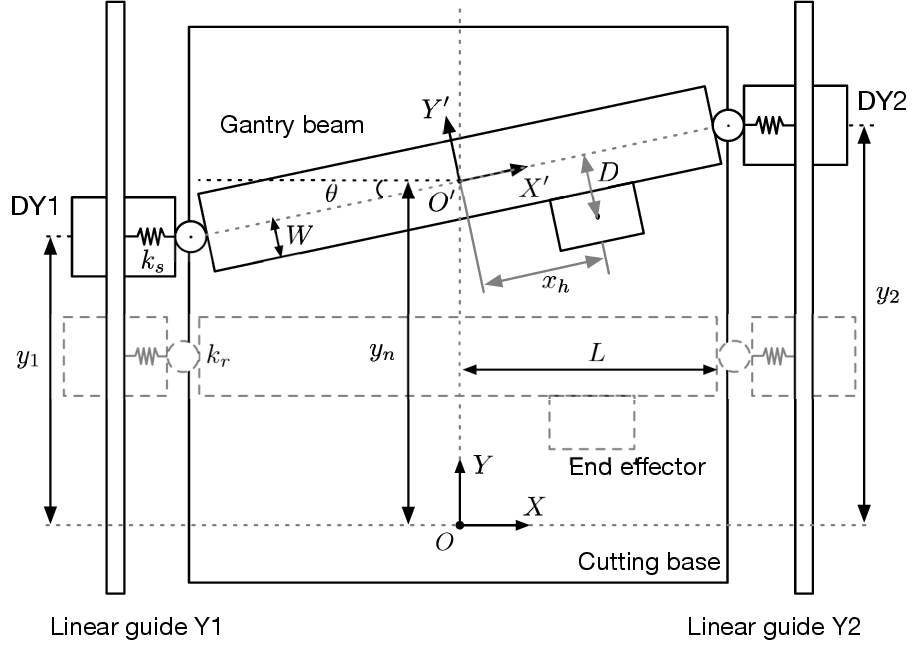}}
	\caption{Schematic diagram of industrial gantry machine.}
	\label{fig:SchemLM}
\end{figure}

The structural parts of machines tend to have more rigidity than joints when finite-element optimisation is applied to the design of industrial machines \cite{wang2015effect}. This motivates the proposition of a physics-based model with linear and torsional springs for systems with limited stiffness joints.

The investigated gantry machine is shown in Fig.~\ref{fig:SchemLM}, demonstrating the dual drives in synchronized and non-synchronized conditions. The torsional springs represent the flexible structural parts that link the drive and crossbeam.  The lateral deformation in the X-axis is applicable with the introduction of linear springs.

The position of end-effector along the X$^{\prime}$ direction is $x_{h}$, which is measurable using linear encoder. The position of dual drives DY1 and DY2 on Y-axis is denoted by $y_{1}$ and $y_{2}$. The centre of gantry beam is $y_{n} = (y_{1}+y_{2})/2$ and the angle of rotation is approximated by $\theta = (y_{2} - y_{1})/2L$ based on small angle approximation. The parameters for system of interest are summarised in Table~\ref{tab:para_sys}.

\begin{table}[tb]
	\caption{Nomenclature of dual drive system}\label{tab:para_sys}
	\centering{}
	\begin{tabular}{ccc}
		\hline
		Symbol & Description & Unit\tabularnewline
		\hline
		$M_{1}$ & Mass of drive DY1 & kg\tabularnewline
		$M_{2}$ & Mass of drive DY2 & kg\tabularnewline
		$M_{e}$ & Mass of end-effector & kg\tabularnewline
		$M_{n}$ & Mass of gantry beam & kg\tabularnewline
		$L$ & Half length of beam & m\tabularnewline
		$W$ & Half width of beam & m\tabularnewline
		$D$ & Projection distance from end-effector onto X$^{\prime}$ & m\tabularnewline
		$k_{x}$, $k_{y}$ & Force constants of X, Y-axis motors & N/A\tabularnewline
		$b_{x}$ & Viscous friction coefficient on gantry & Ns/m\tabularnewline
		$b_{y}$ & Viscous friction coefficient on linear guides & Ns/m\tabularnewline
		$k_{r}$ & Torsional spring constant & N/m\tabularnewline
		$k_{s}$ & Linear spring constant & N/m\tabularnewline
		\hline
	\end{tabular}
\end{table}

By choosing the generalised coordinate as $s = \{x_{h}, y_{n}, \theta\}$, the end-effector position $(x_{e},y_{e})$ is represented as:
\begin{align}
	\label{eq:xe_position}
	x_{e} & = x_{h} \cos\theta + D \sin\theta, \\
	\label{eq:ye_position}
	y_{e} & = y_{n} + x_{h} \sin\theta -D \cos\theta.
\end{align}

The derivation of system dynamics is conducted based on Lagrangian modelling as \cite{yuan2019modelling}:
\begin{equation}
	\frac{d}{dt}\left(\frac{\partial\mathcal{L}}{\partial\dot{s}}\right)-\frac{\partial\mathcal{L}}{\partial s}=Q_{s},
\end{equation}
where $\mathcal{L}$ is the Lagrangian function; $Q_{s}$ are the external generalised forces.

Denote the current of end-effector and dual-drive motors as $i_{x}$, $i_{1}$ and $i_{2}$, respectively. The governing equations of the mechanism are given as:
\begin{equation}
	\label{eq:eqn_xh}
	M_{e}(-x_{h}\dot{\theta}^{2}+\ddot{x}_{h}+D\ddot{\theta}+\ddot{y}_{n}\sin\theta)=k_{x}i_{x}-b_{x}\dot{x}_{h},
\end{equation}
\begin{multline}
	\label{eq:eqn_yn}
	M_{e}\sin\theta\ddot{x}_{h}+M_{t}\ddot{y}_{n}+\Gamma \ddot{\theta}+M_{e}(-x_{h}\sin\theta\dot{\theta}^{2}+2\dot{x}_{h}\cos\theta\dot{\theta} \\
	+ D \dot{\theta}^{2}\cos \theta)+M_{d}L\dot{\theta}^2\sin\theta=k_{y}(i_{1}+i_{2})-b_{y}(\dot{y}_{1}+\dot{y}_{2}),
\end{multline}
\begin{multline}
	\label{eq:eqn_theta}
	M_{e}D\ddot{x}_{h}+\Gamma\ddot{y}_{n}+\Lambda\ddot{\theta}+2k_{r}\theta+2L^{2}k_{s}\sin\theta\left(1-\cos\theta\right) \\
	+2M_{e}\dot{\theta}\dot{x}_{h}x_{h} = \left(k_{y}(i_{2}-i_{1})-b_{y}(\dot{y}_{2}-\dot{y}_{1})\right)L\cos\theta,
\end{multline}
where $M_{d} = M_{1} - M_{2}$ is the mass difference; $M_{t} = M_{1}+M_{2} + M_{e}+M_{n}$ is the total mass of system; $\Gamma = M_{e}D\sin\theta-M_{d}L\cos\theta+M_{e}x_{h}\cos\theta$ and $\Lambda = \left(M_{1}+M_{2}\right)L^{2}+M_{n}(L^{2}+W^{2})/3+M_{e}\left(D^{2}+x_{h}^{2}\right)$ are lumped terms.

The measured and model-predicted frequency-response comparisons for this high-fidelity model are reported in \cite{yuan2019modelling}.

\subsection{Control oriented model}

To derive the control-oriented model for the X-axis movement, the dynamics \eqref{eq:eqn_xh} is reformulated by lumping all terms involving $y_{n}$, $\theta$ into a disturbance term as:

\begin{equation}
	\label{eq:XContrModel}
	\ddot{x}_{h}+\frac{b_{x}}{M_{e}}\dot{x}_{h}=\frac{k_{x}}{M_{e}}i_{x}+d_{x},
\end{equation}
where the disturbance term is $d_{x}\triangleq x_{h} \dot{\theta}^2-d\ddot{\theta}-\ddot{y}_{n}\sin\theta$.

To simplify the Y-axis controller design and eliminate the vibration, the dynamics \eqref{eq:eqn_yn} and \eqref{eq:eqn_theta} are approximated at different operating points $\bar{x}_{h}$ and rearranged as follows by applying a small angle approximation:

\begin{equation}
	\label{eq:YContrModel}
	M_{t}\ddot{y}_{n}+2b_{y}\dot{y}_{n}+M_{e}\bar{x}_{h}\ddot{\theta}=k_{y}\left(i_{1}+i_{2}\right)+d_{1},
\end{equation}
\begin{equation}
	\label{eq:ThetaContrModel}
	M_{e}\bar{x}_{h}\ddot{y}_{n}+\Lambda(\bar{x}_{h})\ddot{\theta}+2b_{y}L^{2}\dot{\theta}+2k_{r}\theta=k_{y}L\left(i_{2}-i_{1}\right)+d_{2},
\end{equation}
where reasonable assumption $M_{1}=M_{2}$ holds; the lumped nonlinear terms are $d_{1}=M_{e}(\bar{x}_{h}\theta\dot{\theta}^{2}-\theta\ddot{x}_{h}-D\theta\ddot{\theta}-2\dot{x}_{h}\dot{\theta}-D\dot{\theta}^{2})$ and $d_{2}=-M_e(2\dot{\theta}\dot{x}_{h}\bar{x}_{h}+D\theta\ddot{y}_{n}+D\ddot{x}_{h})$.

The switched LTI linearisation is adopted because the Y-axis dynamics depend on $x_h$, making a single LTI model insufficient. Compared with nonlinear MPC, which requires solving non-convex optimisation problems at each time step, the switched LTI approach enables the use of efficient convex optimisation for real-time implementation. In \cite{yuan2026contouring}, a switched linear system framework with control invariant (CI) sets was developed for contouring control with structural flexibility. The present work differs by employing robust CI sets to handle bounded disturbances, incorporating input delay into the model and controller design, and providing hardware-in-the-loop experimental validation.

With sampling time $T_{s}$, the discrete-time state space model for the system is derived from~\eqref{eq:XContrModel}, \eqref{eq:YContrModel} and \eqref{eq:ThetaContrModel} as:
\begin{equation}
	\label{eq:xmSSModel}
	\xi_{x}(k+1)=A_{x}\xi_{x}(k)+B_{x}i_{x}(k)+E_{x}d_x(k),
\end{equation}
\begin{equation}
	\label{eq:YnThetaSSModel}
	\xi_{y}(k+1)=A_{y}(\bar{x}_{h})\xi_{y}(k)+B_{y}(\bar{x}_{h})i_{y}(k)+E_{y}(\bar{x}_{h})d_{y}(k),
\end{equation}
where $\xi_{x}\triangleq \left(x_{h},\dot{x}_{h}\right)$,  $\xi_{y}\triangleq ( y_{n}, \dot{y}_{n}, \theta, \dot{\theta} )$ are the states, $i_{x}$, $i_{y} \triangleq \left( i_{1}, i_{2} \right)$ are the inputs; $d_{y} \triangleq \left(d_{1}, d_{2}\right)$ is the lumped nonlinearity vector. The system coefficient matrices of \eqref{eq:xmSSModel} and \eqref{eq:YnThetaSSModel} can be inferred from \eqref{eq:XContrModel}, \eqref{eq:YContrModel} and \eqref{eq:ThetaContrModel}.

The states and inputs of the system are required to stay within constraints as:
\begin{equation}
	\label{eq:state_input_cons}
	\xi_{x}\in\mathcal{X}_{x}\subseteq \mathbb{R}^{2},\,\xi_{y}\in \mathcal{X}_{y}\subseteq \mathbb{R}^{4},\,i_{x}\in\mathcal{U}_{x}\subseteq \mathbb{R},\,i_{y}\in\mathcal{U}_{y}\subseteq \mathbb{R}^{2}.
\end{equation}

With the explicit form of $d_{x}$ and $d_{y}$, the state constraints lead to compact disturbance sets as
\[
d_{x} \in\mathcal{W}_{x}\subseteq\mathbb{R},\,d_{y}\in\mathcal{W}_{y}\subseteq\mathbb{R}^{2}.
\]
In practice, the bounds $\mathcal{W}_{x}$ and $\mathcal{W}_{y}$ are obtained by evaluating the explicit disturbance expressions $d_{x}$, $d_{1}$ and $d_{2}$ over the constrained state operating region $\mathcal{X}_{x}$ and $\mathcal{X}_{y}$, yielding conservative polytopic sets used for the offline RCI computation.

The RCI sets are computed to handle all disturbance realisations within $\mathcal{W}_x$ and $\mathcal{W}_y$. If unmodelled dynamics or parameter uncertainties are present, they can be absorbed into the disturbance sets, and the proposed framework remains effective as long as the total disturbance is bounded within these sets.

Fixed time delay can happen in practical motion control due to communication issue between servo drives and motors. To guarantee the contouring error bound, the state of system should be augmented to explicitly consider the input delay. If $T_{dx}$ step and $T_{dy}$ step input delay exist in X-axis and Y-axis of system, respectively, the dynamics becomes:
\begin{equation}
	\label{eq:Xsys_delay}
	\xi_{x}(k+1) = A_{x}\xi_{x}(k) +B_{x}i_{x}(k-T_{dx})+E_{x}d_{x}(k).
\end{equation}
\begin{equation}
    \label{eq:Ysys_delay}
    \xi_{y}(k+1)=A_{y}(\bar{x}_{h})\xi_{y}(k)+B_{y}(\bar{x}_{h})i_{y}(k-T_{dy})+E_{y}(\bar{x}_{h})d_{y}(k),
\end{equation}

\subsection{Problem statement}

The following model is used to characterise the class of references with operating constraints:
\begin{equation}
	\label{eq:RefModel}
	r_{\alpha}(k+1)=  A_{r}r_{\alpha}(k)+B_{r}u_{\alpha}(k),\,\alpha= \left\{ x,y \right\},
\end{equation}
where $r_{\alpha}$ and $u_{\alpha}$ are the state and input of reference model. The operating constraints are given as:
\begin{equation}
	\label{eq:ref_constraint}
	r_{\alpha} \in \mathcal{X}_{r} \subseteq \mathbb{R}^{n_{r}}, u_{\alpha} \in \mathcal{U}_{r} \subseteq \mathbb{R}
\end{equation}
where $\mathcal{X}_{r}$ and $\mathcal{U}_{r}$ are the compact sets for state and input constraints. For reference with position, velocity and acceleration constraints, the discrete-time system matrices $A_r$ and $B_r$ with sampling time $T_{s}$ are approximated based on Euler forward method as:
\[
A_{r}=\left[\begin{array}{cc}
	1 & T_{s}\\
	0 & 1
\end{array}\right],\,B_{r}=\left[\begin{array}{c}
	0\\
	T_{s}
\end{array}\right].
\]

For X-axis reference, $r_{x}\triangleq \left( x_{e}^{*},v_{x}^{*} \right)$ stands for the desired position and velocity and $a_{x}$ is the desired acceleration. The state $r_{y}\triangleq \left(y_{e}^{*}, v_{y}^{*}\right)$ represents the Y-axis desired position, velocity and $a_{y}$ is the acceleration.

\begin{remark}
	The model \eqref{eq:RefModel} is used to characterise the reference with operating constraints, which can be represented by control invariant (CI) set. It has to be noticed that model \eqref{eq:RefModel} is not a path planner, but the admissible reference is assumed within the reference CI set computed based on \eqref{eq:RefModel}.
\end{remark}

With the above definition, the problem investigated in this work is represented as:
\begin{problem}
	\label{pro:problem}
	Given a desired contouring error bound $\epsilon_{c}$, for systems with dynamics \eqref{eq:Xsys_delay} and \eqref{eq:Ysys_delay}, and reference characterised by constraint \eqref{eq:ref_constraint}, design a control law $u(k) = (i_x(k),i_y(k))$ such that the contouring error $ \epsilon(k) \leq \epsilon_{c}$ holds for all $k \in \mathbb{Z}_{0+}$.
\end{problem}

\begin{remark}
    Although the proposed method is investigated in the context of a biaxial system, it is worth noting that the control algorithm can be extended to handle contouring errors in Euclidean spaces and systems with higher dimensions. This is because the controller is designed based on a general linear time-varying system model, making it applicable to more complex multi-axis systems.
\end{remark}

\section{Control Architecture} \label{sec:control}

\subsection{Contouring error guarantee via axis-wise tracking bounds}

Conventionally, the contouring control involves the calculation or estimation of the contouring error. The existing algorithms for estimating the contouring error can only provide a lower bound of the actual value \cite{yao2011orthogonal}. A control algorithm designed to bound this estimated error does not necessarily guarantee the actual contouring error tolerance is satisfied. Moreover, the estimation algorithms tend to have requirements for the specific shape of the path. To address the limitations of shape-dependent contouring error estimation in conventional methods, we propose a path-agnostic control framework that guarantees contouring error bounds for arbitrary-shape contours. The key theoretical foundation is provided by the following lemma:

\begin{lemma} \label{lemma:ContVsTrackErr}
	Given the desired tolerance of the contouring error $\epsilon_{c}$, the bounded contouring error $\epsilon(k) \leq \epsilon_{c}$, $\forall k \in \mathbb{Z}_{0+}$ can be guaranteed by bounding the X-axis tracking error $\left\Vert e_{x} \right\Vert_{\infty} \leq \epsilon_{x}$ and Y-axis tracking error $\left \Vert e_{y} \right\Vert_{\infty} \leq \epsilon_{y}$  with $\epsilon_{x}+\epsilon_{y} \leq \epsilon_{c}$, where $\epsilon_{x}$ and $\epsilon_{y}$ are the upper bound of tracking errors on X and Y-axis, respectively, and tracking errors are $e_{x} \triangleq x^{*}_{e} - x_{e}$ and $e_{y} \triangleq y_{e}^{*} - y_{e}$.
\end{lemma}

\begin{proof}
Since the contouring error $\epsilon(k)$ is the shortest distance from the actual position to the desired path, it satisfies $\epsilon(k) \leq e(k)$, where $e(k) \triangleq \left\Vert (e_{x}(k),e_{y}(k)) \right\Vert_{2}$ is the Euclidean tracking error. Applying the standard norm inequality $\left\Vert \cdot \right\Vert_{2} \leq \left\Vert \cdot \right\Vert_{1}$ gives $e(k) \leq |e_{x}(k)|+|e_{y}(k)|$. Taking the supremum over all $k$ yields $|e_{x}(k)|+|e_{y}(k)| \leq \left\Vert e_{x} \right\Vert_{\infty}+\left\Vert e_{y} \right\Vert_{\infty}$. Therefore, the inequality chain $\epsilon(k) \leq \left\Vert e_{x} \right\Vert_{\infty}+\left\Vert e_{y} \right\Vert_{\infty} \leq \epsilon_{x}+\epsilon_{y} \leq \epsilon_{c}$ holds for all $k \in \mathbb{Z}_{0+}$.
\end{proof}

This lemma ensures path-agnostic performance, as the tracking error bounds $\epsilon_x$ and $\epsilon_y$ are independent of the contour geometry.

\subsection{Controller formulation}

Let $\bar{\xi}_{y}(k) \triangleq (\xi_{y}(k),i_{y}(k-1),\cdots,i_{y}(k - T_{dy}))$ be the state of Y-axis augmented system, to account for the input delay, the $T_{dy}$ step delay is incorporated into the system dynamics \eqref{eq:Ysys_delay} as:
\begin{align}
	\label{eq:Ysys_delay_aug}
	\bar{\xi}_{y}(k+1) & =\left[\begin{array}{ccccc}
		A_{y}(\bar{x}_{h}) & 0 & \cdots & 0 & B_{y}(\bar{x}_{h})\\
		0 & 0 & \cdots & 0 & 0\\
		0 & 1 & \cdots & 0 & 0\\
		\vdots &  & \ddots &  & \vdots\\
		0 & 0 & \cdots & 1 & 0
	\end{array}\right]\bar{\xi}_{y}(k) \nonumber \\
    & +\left[\begin{array}{c}
		0\\
		1\\
		0\\
		\vdots\\
		0
	\end{array}\right]i_{y}(k) +\left[\begin{array}{c}
		E_{y}(\bar{x}_{h})\\
		0\\
		0\\
		\vdots\\
		0
	\end{array}\right]d_{y}(k) \nonumber\\
	& \triangleq\bar{A}_{y}(\bar{x}_{h})\bar{\xi}_{y}(k)+\bar{B}_{y}i_{y}(k)+\bar{E}_{y}(\bar{x}_{h})d_{y}(k)
\end{align}

Then, the design of $i_{y}$ focuses on achieving $\left\Vert e_{y} \right\Vert_{\infty} = \left \Vert y_{e}^{*} - y_{n}-\bar{x}_{h}\theta - D \right \Vert_{\infty} \leq \epsilon_{y}$ based on \eqref{eq:Ysys_delay_aug}. The dependence on the dynamics $x_{h}$ is avoided by multiple linearisation at different $\bar{x}_{h}$. The resulted switched LTI model decouples and simplifies the controller design between axes. However, for $\left\Vert e_{x}\right \Vert_{\infty} = \left\Vert x_{e}^{*}-x_{h}-D\theta \right \Vert \leq \epsilon_{x}$, the state $\theta$ appears in the control objective for X-axis movement. The Remark~\ref{rem:XTrackBound} is presented to ensure $\left \Vert e_{x} \right \Vert_{\infty} \leq \epsilon_{x}$ with extra constraints enforced when designing $i_{y}$.

\begin{remark}\label{rem:XTrackBound}
    Given a desired X-axis tracking error bound $\epsilon_{x}$ that the system is initialized within, the tracking error $\left \Vert e_{x} \right\Vert_{\infty} \leq \epsilon_{x}$ can be guaranteed if $\left \Vert x_e^{*}-x_{h}\right \Vert_{\infty} \leq \bar{\epsilon}_{x} = \epsilon_{x}-D\theta_{\max}$ is satisfied, where $\theta_{\max}$ is the constraint on rotation angle and the value is chosen to satisfy $\left \Vert \theta \right \Vert_{\infty} \leq \theta_{\max}\leq \epsilon_{x}/D$.
\end{remark}

\begin{remark}
    The $D\theta$ term in the X-axis tracking error introduces an asymmetry between the X and Y-axis error bound. A tightening value of $\theta_{\max}$ leads to a loose constraint $\bar{\epsilon}_{x}$ for the X-axis controller but imposes a more strict tolerance $\epsilon_{y}$ when designing control input $i_{y}$.
\end{remark}

It is noted that the small-angle approximation is applied only in the control-oriented model \eqref{eq:YContrModel}--\eqref{eq:ThetaContrModel}, while the full nonlinear dynamics \eqref{eq:eqn_xh}--\eqref{eq:eqn_theta} retain the complete trigonometric terms. The validity of this approximation is ensured by the hard constraint $\left\Vert \theta \right\Vert_{\infty} \leq \theta_{\max}$ enforced via the RCI sets so that the system always operates within the valid linearisation range.

The controller design starts with the control law $i_{y}$ to ensure the following operating constraint:
\begin{equation}
	\label{eq:y_err_bound}
	 \left \Vert y_{e}^{*} - y_{e} \right \Vert_{\infty} \leq \epsilon_{y}, \, \left \Vert \theta \right\Vert_{\infty} \leq \theta_{\max}
\end{equation}

The control-oriented model \eqref{eq:Ysys_delay_aug} is assumed approximated at different operating points as $\bar{x}_{h}^{j} = \{ \bar{x}_{h}^{1},\cdots, \bar{x}_{h}^{N_{l}}\}$, for $j \in \mathbb{Z}_{[1,N_{l}]}$, where $N_{l}$ is the total number of linearisation points. Let $\Delta x_{h}$ be the interval of the system been linearised, the actual position $x_{h}$ falls in the range between linearisation points as $x_{h}=\bar{x}_{h}^{j} + \delta x_{h}, \; j \in \mathbb{Z}_{[1,N_{l}]}$ with $|\delta x_{h}|\leq \Delta x_{h}$.

Following the algorithm in \cite{yuan2019error} to achieve error bounded tracking for system with disturbance, the reference characterised by constraint \eqref{eq:ref_constraint} is assumed staying within a CI set $r_{y}\in \mathcal{R}^{r_{y}}_{\infty}$ which can be computed offline based on set iteration using model \eqref{eq:RefModel} and constraints \eqref{eq:ref_constraint}. To ensure the tracking error bound and rotation angle constraints are satisfied at different operating point $\bar{x}_{h}^{j}$, the augmented states $(\bar{\xi}_{y},r_{y})$ are required to stay in a group of RCI sets $\mathcal{R}_{\bar{x}_{h}^{j}}^{\bar{\xi}_{y},r_{y}}$, $j\in \mathbb{Z}_{[1,N_{l}]}$.

Conventional methods for computing the RCI set rely on the iteration of particular sets until two consecutive sets are equal. For systems with external disturbance, the corresponding maximal RCI set is defined by an infinite number of inequalities and the stopping criterion based on two equal sets becomes numerically difficult and intractable. Here, we follow the RCI set computation algorithm in \cite{yuan2019error} and extend it for ensuring angle constraint and bounded tracking error with finite step termination. The RCI sets $\mathcal{R}_{\bar{x}_{h}^{j}}^{\bar{\xi}_{y},r_{y}}$ at different operating points $\bar{x}_{h}^{j}$ are computed based on Algorithm~\ref{alg:RCI_comp}.

\begin{algorithm}
	\caption{RCI set computation for switched LTI system}\label{alg:RCI_comp}
	\begin{algorithmic}[1]
		\State \emph{Initialization:}
		\State $m \gets 0$, $\mathcal{R}_{0} \gets \mathcal{R}_{s} \cap \bar{\mathcal{R}}_{0},$ where
		\State $\mathcal{R}_{s} \gets \left\{ (\bar{\xi}_{y}, r_{y}) \in \mathbb{R}^{6+2T_{dy}} \, | \, \bar{\xi}_{y} \in \mathbb{R}^{4+2T_{dy}}, \, r_{y} \in \mathcal{R}^{r_{y}}_{\infty} \right\}$,
		\Statex $\bar{\mathcal{R}}_{0} \gets \{ (\bar{\xi}_{y}, r_{y}) \in \mathbb{R}^{6+2T_{dy}} \,| \, \bar{\xi}_{y} \in \mathcal{X}_{y} \times \mathcal{U}_{y}^{T_{dy}}, \left\Vert \theta \right\Vert_{\infty} \leq \theta_{\max}$, $ \left\Vert y_{e}^{*} - y_{N} - \bar{x}_{h}^{j}\theta +D \right\Vert_{\infty} \leq \epsilon_{y} \}$, $j = 1,\cdots, N_{l}$.
		\State \emph{Iteration:}\label{alg:IterCont}
		\State $\mathcal{R}_{m+1} \gets \mathcal{R}_{s} \cap \bar{\mathcal{R}}_{m+1}$, where $\bar{\mathcal{R}}_{m+1} \gets \hat{P} \left( \bar{\mathcal{R}}_{m} ,\rho \right) \cap \bar{\mathcal{R}}_{m},$
		\begin{multline*}
			\hat{P}  \left( \bar{\mathcal{R}}_{m},\rho \right) \gets \{ (\bar{\xi}_{y}, r_{y}) \in \mathbb{R}^{4+2T_{dy}} \,|\, \exists i_{y} \in  \mathcal{U}_y, \nonumber   \\
			\Bigl( \bar{A}_y(\bar{x}_{h}^{j})\bar{\xi}_{y} + \bar{B}_y i_{y} + \bar{E}_{y}(\bar{x}_h)d_{y}, A_{r}r_{y} + B_{r} u_{y} \Bigr) \in \bar{\mathcal{R}}_{m}, \nonumber \\
			 - \mathbb{B}^{6}(\rho) \forall (d_{y} \times u_{y}) \in \left(\mathcal{W}_{y} \times \mathcal{U}_{r} \right) \}.
		\end{multline*}
		\If{$\mathcal{R}_{m}  -  \mathbb{B}^{6} (\rho) \subseteq \mathcal{R}_{m+1}$}
		\State \Return $\mathcal{R}_{\bar{x}_{h}^{j}}^{\bar{\xi}_{y},r_{y}} \gets \mathcal{R}_{m+1}$
		\Else
		\State $m \gets m+1$
		\State \textbf{go to} \ref{alg:IterCont}
		\EndIf
	\end{algorithmic}
\end{algorithm}

\begin{remark}
	There is a trade-off in that using finer intervals $\Delta x_{h}$ means a better approximation of $x_h$, but leads to more controller switching and requires more off-line calculation and storage of the RCI sets.
\end{remark}

Algorithm~\ref{alg:RCI_comp} computes the RCI sets offline for the Y-axis linearised models. For each linearisation point $j\in\mathbb{Z}_{[1,N_l]}$, let $M_j$ be the number of set-iteration updates until termination. Since each update is carried out in the augmented space $(\bar{\xi}_{y},r_{y})\in\mathbb{R}^{6+2T_{dy}}$, the offline complexity is $\mathcal{O}\!\left(\sum_{j=1}^{N_l} M_j\right)$ set updates with fixed problem size per update. The X-axis case is obtained in the same way.

Based on the feedback $x_{h}$ at time instant $k$, the value of $j$ and the linearisation point $\bar{x}_{h}^{j}$ are determined following $\bar{x}_{h}^{j} - \frac{\Delta x_{h}}{2} \leq x_{h} \leq \bar{x}_{h}^{j} + \frac{\Delta x_{h}}{2}$. With $N \in \mathbb{Z}_{0+}$ step future reference $\gamma_{y}^{N}(k)=\left(r_{y}(k),\cdots,r_y(k+N-1)\right)$, the predictive controller solves the below optimisation:
\begin{align}
	U_{y}^{*}\left(k\right) =& \arg\min_{U_{y}(k)}\sum_{i=0}^{N-1}\Bigl(Q_{y}\left(y_{e}^{*}(k+i)-y_{e(i|k)}\right)^{2} \nonumber \\
	& \qquad \qquad + \left\Vert i_{y(i|k)}\right\Vert^2_{R_y} \Bigr)\nonumber \\
	s.t. \;\;& \bar{\xi}_{y(i+1|k)}=\bar{A}_{y}(\bar{x}_{h}^{j})\bar{\xi}_{y(i|k)}+\bar{B}_{y}i_{y(i|k)},\nonumber \\
	& y_{e(i|k)}=\left[\begin{array}{ccccc}
		1 & 0 & \bar{x}_{h}^{j} & 0 & 0_{1,2T_{dy}}\end{array}\right]\bar{\xi}_{y(i|k)}-D,\nonumber \\
	& y_{e}^{*}(k+i)=\left[\begin{array}{cc}
		1 & 0\end{array}\right]r_{y}(k+i), \nonumber \\
	& \bar{\xi}_{y(i|k)}\in\mathcal{R}^{\bar{\xi}_{y},r_{y}}_{\bar{x}_{h}^{j}}(\bar{\xi}_{y},r_{y}(k+i)),\, i\in\mathbb{Z}_{[0,N-1]}, \nonumber \\
	& \qquad \qquad \qquad j\in \mathbb{Z}_{[1,N_l]}, \nonumber \\
	& \bar{\xi}_{y(0|k)}=\bar{\xi}_{y}(k),
	\label{eq:MPC_Y_Theta}
\end{align}
where $\left\Vert i_{y(i|k)} \right\Vert^2_{R_y}=i_{y(i|k)}^TR_y i_{y(i|k)}$, $Q_{y}$ and $R_{y}$ are the tuning weights on the Y-axis position tracking error and control input, respectively. The optimal control input set is $U_{y}^{*}(k) \triangleq \left(i_{y(0|k)}^{*}, \cdots, i_{y(N-1|k)}^{*}\right)$. At each time instant, the optimal control $i_{y}(k) = i_{y(0|k)}^{*}$ is applied to the plant. The constraint of rotation angle $\theta$, that is required in Remark~\ref{rem:XTrackBound}, is guaranteed by enforcing the state $\bar{\xi}_{y}$ within the calculated RCI set.

The state of reference for X-axis is assumed within the control invariant set $r_{x} \in \mathcal{R}_{\infty}^{r_{x}}$. The design of control law for $i_{x}$ is to ensure the following error bound:
\begin{equation}
	\label{eq:x_err_bound}
	\left\Vert x_{e}^{*}-x_{h}\right\Vert_{\infty} \leq \bar{\epsilon}_{x}
\end{equation}

Then, let $\bar{\xi}_{x}(k) \triangleq (\xi(k),i_{x}(k-1),\cdots,i_{x}(k-T_{dx}))$ be the state of X-axis augmented system, to account for the input delay, the $T_{dx}$ step delay is incorporated into the system dynamics \eqref{eq:Xsys_delay} as:

\begin{align}
	\label{eq:Xsys_delay_aug}
	\bar{\xi}_{x}(k+1) & =\left[\begin{array}{ccccc}
		A_{x} & 0 & \cdots & 0 & B_{x}\\
		0 & 0 & \cdots & 0 & 0\\
		0 & 1 & \cdots & 0 & 0\\
		\vdots &  & \ddots &  & \vdots\\
		0 & 0 & \cdots & 1 & 0
	\end{array}\right]\bar{\xi}_{x}(k)+\left[\begin{array}{c}
		0\\
		1\\
		0\\
		\vdots\\
		0
	\end{array}\right]i_{x}(k) \nonumber\\
	& +\left[\begin{array}{c}
		E_{x}\\
		0\\
		0\\
		\vdots\\
		0
	\end{array}\right]d_{x}(k) \nonumber\\
	& \triangleq\bar{A}_{x}\bar{\xi}_{x}(k)+\bar{B}_{x}i_{x}(k)+\bar{E}_{x}d_{x}(k)
\end{align}

 To ensure the constraint \eqref{eq:x_err_bound} is satisfied, the Algorithm~\ref{alg:RCI_comp} is modified with $\mathcal{R}_{s} \allowbreak = \allowbreak \left \{ \allowbreak (\bar{\xi}_{x}, r_{x}) \in \mathbb{R}^{6+2T_{dx}} \, | \, \bar{\xi}_{x} \in \mathbb{R}^{4+2T_{dx}}, \, r_{x} \in \mathcal{R}^{r_{x}}_{\infty} \right\}$ and $\bar{\mathcal{R}}_{0} = \allowbreak \left\{ (\bar{\xi}_{x}, r_{x}) \allowbreak \in \mathbb{R}^{6+2T_{dy}} \,| \, \bar{\xi}_{x} \in \mathcal{X}_{x}, \, \left \Vert x_{e}^{*} - x_{h} \right \Vert _{\infty} \leq \bar{\epsilon}_{x} \right\}$ to initialize the iteration and the following set computation is utilised for the RCI set $\mathcal{R}^{r_{x}}_{\infty}$ computation:
\begin{multline*}
	\hat{P}  \left( \bar{\mathcal{R}}_{m},\rho \right) \gets \{ (\bar{\xi}_{x}, r_{x}) \in \mathbb{R}^{6+2T_{dx}} \,|\, \exists i_{x} \in  \mathcal{U}_x, \nonumber \\
	\left( \bar{A}_{x} \bar{\xi}_{x} + \bar{B}_{x}i_{x}+\bar{E}_{x}d_{x}, A_{r}r_{x} + B_{r} u_{x} \right) \in \bar{\mathcal{R}}_{m} - \mathbb{B}^{6+2T_{dx}}(\rho), \nonumber \\
	\forall (d_{x} \times u_{x}) \in \left(\mathcal{W}_{x} \times \mathcal{U}_{r} \right) \}.
\end{multline*}

Based on the computed RCI set $\mathcal{R}^{\bar{\xi}_{x},r_{x}}$ and $N$ steps reference trajectory, i.e., $\gamma_{x}^{N}(k)= \left( r_{x}(k),\cdots,r_{x}(k+N-1) \right) $, the control law $i_{x}$ is computed by solving the following optimisation problem:
\begin{align}
	U_{x}^{*}\left(k\right) =& \arg\min_{U_{x}(k)}\sum_{i=0}^{N-1}\Bigl(Q_{x}\left(x_{e}^{*}(k+i)-x_{h(i|k)}\right)^{2} \nonumber \\
	& \qquad \qquad +R_{x}i_{x(i|k)}^{2} \Bigr)\nonumber \\
	s.t. \;\;& \bar{\xi}_{x(i+1|k)}=\bar{A}_{x}\bar{\xi}_{x(i|k)}+\bar{B}_{x}i_{x(i|k)},\nonumber \\
	& x_{h(i|k)}=\left[\begin{array}{cc}
		1 & 0\end{array}\right]\bar{\xi}_{x(i|k)},\nonumber \\
	& x_{e}^{*}(k+i)=\left[\begin{array}{cc}
		1 & 0\end{array}\right]r_{x}(k+i), \nonumber \\
	& \bar{\xi}_{x(i|k)}\in\mathcal{R}^{\bar{\xi}_{x},r_{x}}(\bar{\xi}_{x},r_{x}(k+i)),\, i \in\mathbb{Z}_{[0,N-1]}, \nonumber \\
	& \bar{\xi}_{x(0|k)}=\bar{\xi}_{x}(k),
	\label{eq:MPC_X}
\end{align}
where $Q_{x}$ and $R_{x}$ are the cost function weight on X-axis position error and control input respectively; $U_{x}^{*}(k) = \left( i_{x(0|k)}^{*}, \cdots, i_{x(N-1|k)}^{*} \right)$. At each time instant, the optimal control input $i_{x}(k) \triangleq i_{x(0|k)}^{*}$ is applied to the plant.

\subsection{Analysis of closed-loop property}

The following theorem presents the recursive feasibility of the proposed control law \eqref{eq:MPC_Y_Theta} and  \eqref{eq:MPC_X}.

\begin{theorem}
    \label{the:feasibility}
    Consider the reference subject to \eqref{eq:ref_constraint}, system \eqref{eq:xmSSModel} and \eqref{eq:YnThetaSSModel} subject to \eqref{eq:state_input_cons}, and operating constraint \eqref{eq:y_err_bound} and \eqref{eq:x_err_bound} hold. If the control law \eqref{eq:MPC_Y_Theta} and \eqref{eq:MPC_X} is feasible at time instant $k\in \mathbb{Z}_{0+}$, then they are feasible for all $\bar{k}$, $\bar{k} \geq k$, and the closed-loop system satisfies \eqref{eq:state_input_cons}, \eqref{eq:y_err_bound} and \eqref{eq:x_err_bound}.
\end{theorem}
\begin{proof}
For the recursive feasibility of \eqref{eq:MPC_X}, it means given the optimal solution $U_{x}^{*}(k) = ( i_{x(0|k)}^{*}, \cdots, i_{x(N-1|k)}^{*} )$ at any time instant $k\in \mathbb{Z}_{0+}$, there exists a feasible solution $\tilde{U}_{x}(k+1) = (\tilde{i}_{x(0|k+1)}, \cdots, \tilde{i}_{ x(N-1|k+1) })$ at time $k+1$.

At time instant $k$, $i_{x}(k) = i_{x(0|k)}^{*}$ applies to the plant and $(\bar{\xi}_{x(1|k)}, r_{x}(k+1)) \in \mathcal{R}^{\bar{\xi}_{x},r_{x}}$ holds for every $a_{x}(k)\in \mathcal{U}_{r}$. Since $r_{x}(k) \in \mathcal{X}_{r}$, at $k+1$ time instant $(\bar{\xi}_{x}(k+1),r(k+1))\in \mathcal{R}^{\bar{\xi}_{x},r_{x}}$ holds, which means there exists $\tilde{i}_{x(0|k+1)}$ such that $(\bar{\xi}_{x(1|k+1)}, r(k+1+1))\in \mathcal{R}^{\bar{\xi}_{x},r_{x}}$ for $r(k+1) \in \mathcal{X}_{r}$. Since reference $r(k+1+i)$ for $i = 0,\cdots, N-1$ satisfies the reference constraint \eqref{eq:ref_constraint}, control input $\tilde{U}_{x}(k+1)$ can be constructed and $\tilde{U}_{x}(k+1)$ is feasible. The proof for \eqref{eq:MPC_Y_Theta} is similar and is thus omitted.
\end{proof}

\begin{remark}
    \label{coro:solution}
    Under the condition of Theorem~\ref{the:feasibility}, the control input $u(k) = (i^{*}_{x(0|k)}, i^{*}_{y(0|k)})$ for $k \in \mathbb{Z}_{0+}$ is the solution of Problem~\ref{pro:problem}.
\end{remark}

\begin{remark}

The Theorem~\ref{the:feasibility} and Remark~\ref{coro:solution} hold for systems with input delay if the augmented system is used in RCI set computation and MPC formulation.
\end{remark}

\subsection{Procedure of implementing proposed algorithm}

The control architecture consists of two main stages, as illustrated in Fig.~\ref{fig:Schematic}. In the offline stage, the RCI sets are computed for each linearisation point using the switched LTI models and the disturbance sets. In the online stage, the feedback $x_h$ determines the active linearisation point, the corresponding RCI set is selected, and the MPC optimisation is solved to compute the control inputs $i_x$ and $i_y$ while enforcing the RCI set constraints.

\begin{figure*}[t]
    \centerline{\includegraphics[width=0.6\textwidth]{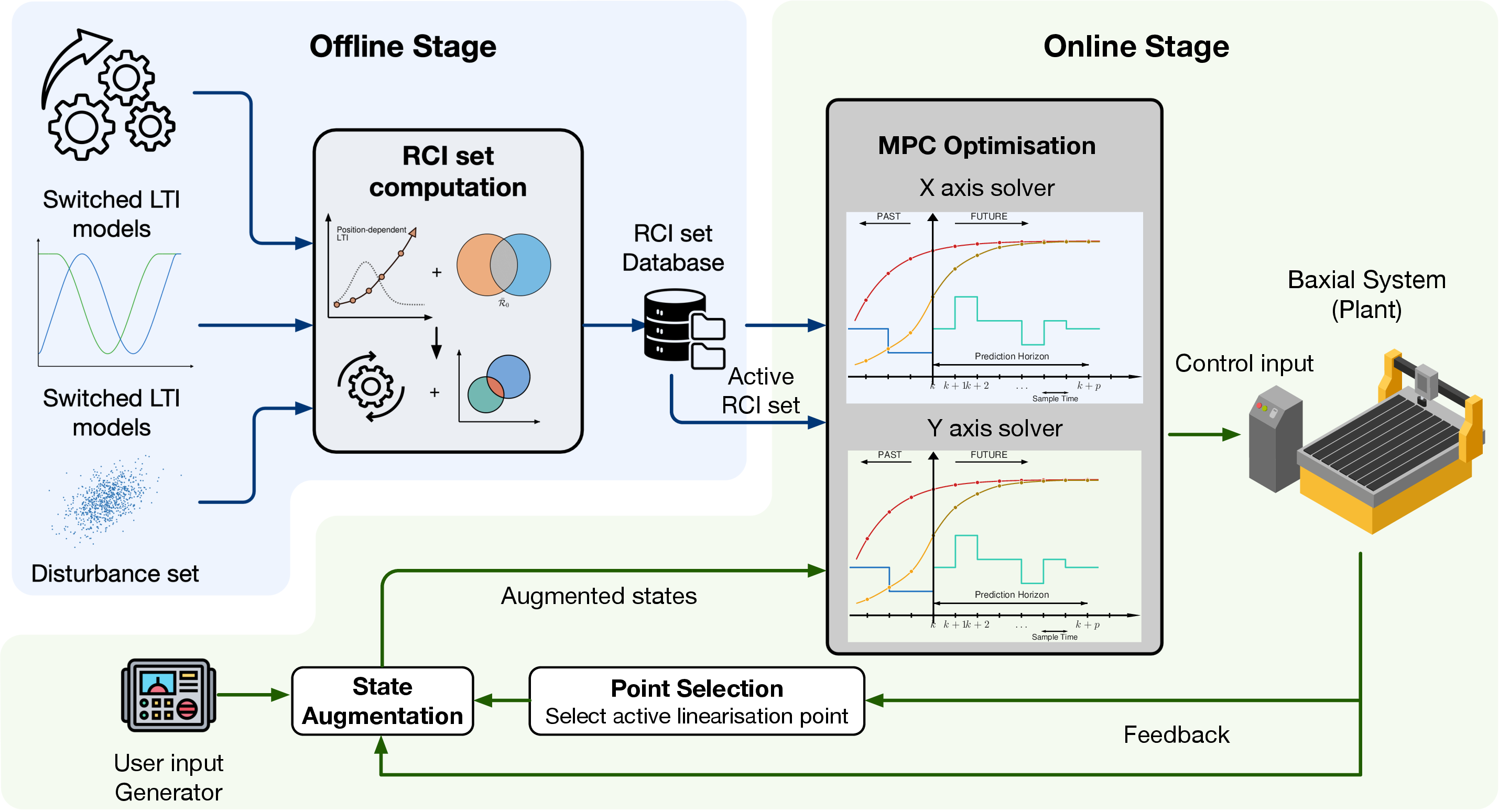}}
    \caption{Block diagram of the proposed contouring error-bounded control architecture.}
    \label{fig:Schematic}
\end{figure*}

To obtain a clear overview of implementing the proposed control algorithm, the procedure is summarised as:

\subsubsection*{Offline}

\begin{enumerate}
	\item Choose the value of tracking error bound $\epsilon_{x}$, $\epsilon_{y}$ based on the given contouring error bound $\epsilon_{c}$ and $\epsilon_{x}+\epsilon_{y}\leq \epsilon_{c}$.
	\item Choose the value of the maximum rotation angle constraint $\theta_{max}$ based on $\theta_{max} \leq \epsilon_{x}/D$ and compute the $\bar{\epsilon}_{x}$ using $\bar{\epsilon}_{x} = \epsilon_{x} - D\theta_{max}$.
	\item Determine the linearisation point $\bar{x}_{h}^{j}$, $j \in \mathbb{Z}_{[1,N_{l}]}$ based on the operation range $\xi_{x}\in \mathcal{X}_{x}$.
	\item Obtain the augmented control-oriented models as shown in \eqref{eq:Ysys_delay_aug} and \eqref{eq:Xsys_delay_aug}.
	\item Compute the reference CI set $\mathcal{R}^{r_{x}}_{\infty}$ and $\mathcal{R}^{r_{y}}_{\infty}$ based on the requirements of path been traced, i.e., sets of $\mathcal{X}_{r}$ and $\mathcal{U}_{r}$.
	\item Compute the RCI sets $\mathcal{R}^{\bar{\xi}_{x},r_{x}}$ and $\mathcal{R}_{\bar{x}_{h}^{j}}^{\bar{\xi}_{y},r_{y}}$ using the proposed algorithm.
\end{enumerate}

\subsubsection*{Online}

\begin{enumerate}
	\item Obtain the current states $\bar{\xi}_{x}(k)$ and $\bar{\xi}_{y}(k)$ based on feedback $x_{h}(k)$, $y_{1}(k)$ and $y_{2}(k)$.
	\item Determine the value of $j$ in $\bar{x}_{h}^{j}$ based on the feedback of $x_{h}$, and update the $\mathcal{R}_{\bar{x}_{h}^{j}}^{\bar{\xi}_{y},r_{y}}$ in \eqref{eq:MPC_Y_Theta}.
	\item Compute $U_{y}^{*}(k)$ and $U_{x}^{*}(k)$ using \eqref{eq:MPC_Y_Theta} and \eqref{eq:MPC_X} respectively.
	\item Implement control input $u(k) = (i_{x(0|k)}^{*}, i_{y(0|k)}^{*})$.
\end{enumerate}

\section{Results} \label{sec:result}

Ideally, the contouring performance of the proposed controller should be tested on the laser machine described in Section~\ref{sec:prob_formu}. However, due to positional discrepancies between the actuators and the end-effector, and the lack of direct positional feedback from the end-effector, the performance can only be evaluated indirectly through post-process product quality.

To address this limitation, a platform is designed that replicates the biaxial system dynamics while incorporating input delay and discretisation conditions. This section details the design of the test bench, the replication of the biaxial system, and the performance validation results of the proposed controller.

\subsection{System identification}
To parameterize the proposed model, data were collected from a commercial laser machine with a moving distance of $1.5$ m and $3$ m for the X and Y-axis, respectively. Analogue incremental encoders ERN1387 with 2048 line resolution are used to measure the position of actuators.

A $4$ A peak amplitude chirp-signal current is implemented on the Y-axis motor to excite the system when the X-axis motor is held constant at different positions. The parameters are identified based on the frequency response from the Y-axis current to the dual-drive velocities using the collected data. The physical system exhibits changes in anti-resonance frequency when the X-axis motor moves along the gantry beam, exhibiting a state-dependence of this motor on the resulting dynamics and validated results are demonstrated in \cite{yuan2019modelling}, which is omitted here.

\subsection{Experiment setup and configuration}

\begin{figure}[tbp]
	\centerline{\includegraphics[width = 0.9 \columnwidth]{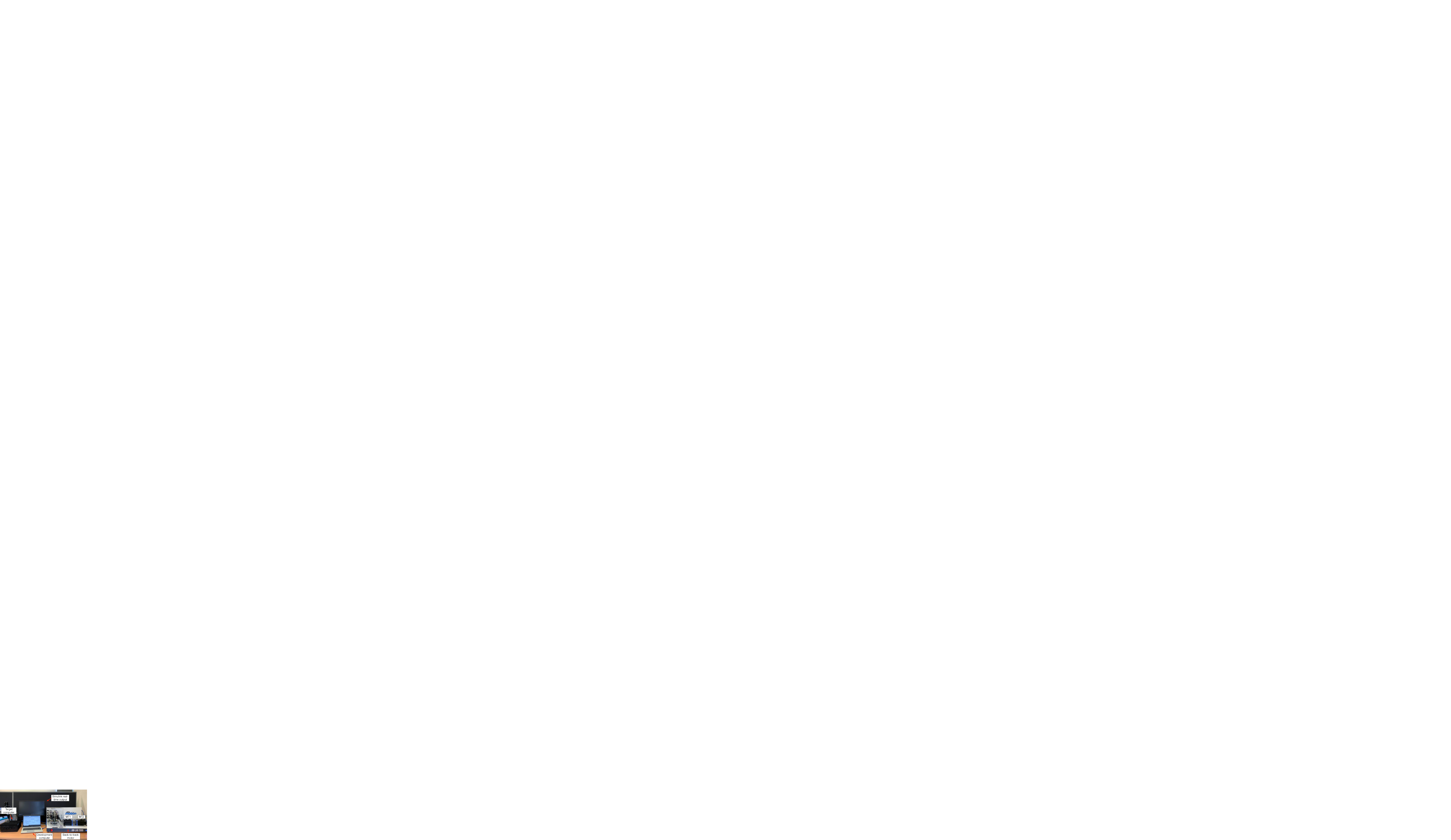}}
	\caption{Experimental test bench for contouring control.}{\label{fig:ContTestBench}}
\end{figure}

To replicate the biaxial system, a hardware-in-the-loop test bench was developed, as illustrated in Fig.~\ref{fig:ContTestBench}. It comprises a development computer, a target computer, and two identical rotary motors arranged in a back-to-back configuration. The controller is implemented online using the Simulink Real-Time environment. The target computer is a Dell Embedded Box PC 5000 equipped with an Intel Core i7-6820EQ processor.

The motors are equipped with two independent servo drives, identical to those used in actual industrial applications. Prior to the deployment of the proposed contouring control algorithms, the current loop response was analysed, and the results are shown in Fig.~\ref{fig:CurrentDelay}. The communication rate between the target PC and the servo drive is first configured at a fast sampling rate with 4~kHz, with a proportional-integral (PI) controller embedded at the drive level for current tracking. Although the servo drive has an intrinsic control cycle of 1~ms, an additional 1~ms delay is observed in the closed-loop response due to communication latency between the target PC and the drive. The controller is therefore configured to operate at a 2~ms update rate, and a one-sample input delay is effectively introduced. This delay is taken into account in the subsequent modelling and control design.

\begin{figure}[tb]
	\centerline{\includegraphics[width = 0.85\columnwidth]{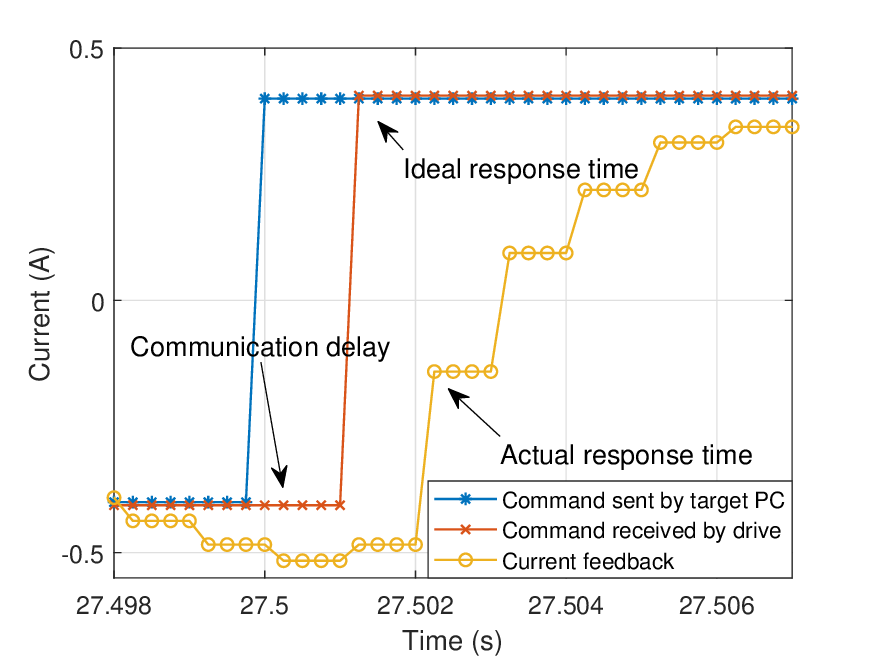}}
	\caption{Closed-loop response of current loop on motor MT1.}{\label{fig:CurrentDelay}}
\end{figure}

By properly generating the load torque for motor MT2 and applying a rotary-to-translational motion conversion, the movement of motor MT1 can represent the dynamics of the $x_{h}$ coordinate. The configuration begins with the dynamics of motor MT1, described as follows:
\begin{equation}
    \label{eq:AMRotary}
    J_{A}\ddot{\theta}_{h}=k_{A}i_{x}-b_{r}\dot{\theta}_{h}-F_{L}-F_{n},
\end{equation}
where $J_{A}$ is the total moment of inertial including the contributions from both motors and the coupling mechanism, $\theta_{h}$ is the rotary angle of motor MT1, $k_{A}$ is the torque constant of the motors, $b_{r}$ is the rotary viscous-friction coefficient, $F_{L}$ is the load torque imposed on motor MT1 with expression $F_{L} = k_{A}i_{T}$ and $i_{T}$ is the current of the motor MT2, $F_{n}$ is the unmodeled nonlinear effects.

With the unit conversion $\theta_{h} \triangleq \frac{x_{h}}{r_{A}}$, where $r_{A}$ is the radius of the motor shaft, the dynamics \eqref{eq:AMRotary} is rewritten as:
\begin{equation}
	\label{eq:AMLinear}
	\ddot{x}_h = \frac{k_{A}}{M_{A}}i_{x} - \frac{b_{A}}{M_{A}} \dot{x}_{h} -\frac{F_{L}}{M_{A}} - \frac{F_{n}}{M_{A}},
\end{equation}
where $M_{A} \triangleq \frac{J_{A}}{r_{A}}$ is the equivalent total mass and $b_{A} \triangleq \frac{b_{r}}{r_{A}}$ is the viscous friction coefficient for the translational movement.

By choosing the torque command of motor MT2 as
\begin{equation}
	\label{eq:MT2icommand}
	F_{L}  = M_{A}\left( (\frac{k_{A}}{M_{A}} - \frac{k_{x}}{M_{e}})i_{x} - ( \frac{b_{A}}{M_{A}}- \frac{b_{x}}{M_{e}})\dot{x}_{h}  - d_{x}\right),
\end{equation}
where $M_{e}$, $k_{x}$ and $b_{x}$ are the parameters of the industrial laser machine, the dynamics of system \eqref{eq:AMLinear} becomes:
\begin{equation}
	\label{eq:SimuXhDynamics}
	\ddot{x}_{h} = \frac{k_{x}}{M_{e}}i_{x} - \frac{b_x}{M_{e}}\dot{x}_{h}+d_{x}-\frac{F_{n}}{M_{A}}.
\end{equation}

With the consideration of one-sample delay in the current loop, the dynamics of the test bench in discrete time with sampling $T_{s}$ becomes:
\begin{align*}
	\dot{x}_{h}(k+1) = & \frac{T_{s}k_{A}}{M_{A}}i_{x}(k-1) + \left( 1-\frac{T_{s}b_{A}}{M_{A}} \right) \dot{x}_{h}(k) \\
	& - \frac{T_{s}F_{L}(k-1)}{M_{A}} - \frac{T_{s}F_{n}(k)}{M_{A}}.
\end{align*}

The delay of current leads to a delayed torque when implementing the command \eqref{eq:MT2icommand} on motor MT2 as :
\begin{align*}
	F_{L}(k-1) = & M_{A}\Bigl( (\frac{k_{A}}{M_{A}} - \frac{k_{x}}{M_{e}})i_{x}(k-1) \\
	& - ( \frac{b_{A}}{M_{A}}- \frac{b_{x}}{M_{e}} )\dot{x}_{h}(k-1)  - d_{x}(k-1) \Bigr).
\end{align*}

Then the dynamics of system is:
\begin{equation}
	\label{eq:XsysDelayDynamics}
	\dot{x}_{h}(k+1) = \frac{T_{s}k_{x}}{M_{e}}i_{x}(k-1) + \left(1-\frac{T_{s}b_{A}}{M_{A}}\right) \dot{x}_{h}(k)+F_{N}(k)
\end{equation}
where $F_{N}(k) = T_{s}\left(\frac{b_{A}}{M_{A}} - \frac{b_{x}}{M_{e}}\right) \dot{x}_{h}(k-1) + T_{s}d_{x}(k-1) - \frac{T_{s}F_{n}(k)}{M_{A}}$ can be considered as the disturbance and is assumed within a bounded set $F_{N} \in \mathcal{W}_{N}$.

With $\xi_{x} \triangleq \left(x_{h}, \dot{x}_{h}\right)$ and $i_x(k-1)$ be the state and input respectively, the dynamics \eqref{eq:XsysDelayDynamics} is represented in state space form as:
\begin{equation}
	\label{eq:XsysDelayss}
	\xi_{x}(k+1)= A_{e}\xi_{x}(k) + B_{e}i_x(k-1)+E_{e}F_{N}(k),
\end{equation}
where the matrix $A_{e}$, $B_{e}$ and $E_{e}$ can be inferred from \eqref{eq:XsysDelayDynamics}.

Following the step in augmenting the system dynamics, let $\bar{\xi}_{x}(k) \triangleq \left( \xi_{x}(k), i_{x}(k-1) \right)$ be the state of the augmented system, the one-sample delay is incorporated into the X-axis system dynamics \eqref{eq:XsysDelayss} as:
\begin{align}
	\label{eq:XsysDelayCont}
	\left[\begin{array}{c}
		\xi_{x}(k+1)\\
		i_{x}(k)
	\end{array}\right]  = & \left[\begin{array}{cc}
		A_{e} & B_{e}\\
		0_{1,2} & 0
	\end{array}\right]\left[\begin{array}{c}
		\xi_{x}(k)\\
		i_{x}(k-1)
	\end{array}\right] \nonumber \\
	& +\left[\begin{array}{c}
		0_{2,1}\\
		1
	\end{array}\right]i_{x}(k)+\left[\begin{array}{c}
		E_{e}\\
		0
	\end{array}\right]F_{N}(k) \nonumber\\
	\triangleq & \bar{A}\bar{\xi}_{x}(k) + \bar{B}i_{x}(k) + \bar{E}F_{N}(k).
\end{align}

This control-oriented model \eqref{eq:XsysDelayCont} and bounded disturbance $F_{N}\in \mathcal{W}_{N}$ are used to compute the RCI set $\mathcal{R}^{\bar{\xi}_x,r_x}$ based on the proposed algorithm. The nominal model $\bar{\xi}_{x}(k+1) = \bar{A}\bar{\xi}_{x}(k)+\bar{B}i_{x}(k)$ is used in the MPC based formulation \eqref{eq:MPC_X} for error bounded tracking.

The back-to-back motor configuration is employed to replicate the dynamics described in \eqref{eq:eqn_xh}, with the nonlinear disturbance $F_{n}$ inherently introduced by this setup. This experimental platform is utilised in conjunction with the simulated dynamics \eqref{eq:eqn_yn} and \eqref{eq:eqn_theta}, which are executed on the target computer to emulate the plant for the biaxial contouring control experiments.

\subsection{Contouring results}

To validate the efficacy of the proposed algorithm, the end-effector is required to follow a contour consisting of a circular path and a straight line with the maximum velocity at $0.1$ m/s and maximum acceleration at $1$ m/s$^2$. The circle has centre at $(0,0)$ with $0.08$ m radius.

The desired contouring error bound $\epsilon_{c}$ is chosen as $4$ mm, equivalent to $5\%$ of the circular radius in the desired trajectory. The controller is configured to update at $T_{s} = 2$ ms and the data sampling rate is chosen as $4$ kHz. The set manipulation is based on MPT3 toolbox \cite{herceg2013multi}. The solver CVXGEN is utilised to ensure the fast computation of the optimisation problem \cite{mattingley2012cvxgen}. The CVXGEN solver is implemented with a maximum of $25$ iterations, an optimality tolerance of $10^{-4}$, and a feasibility residual tolerance of $10^{-6}$.

In order to guarantee the $4$ mm contouring error bound, the upper bound of X and Y-axis tracking error is chosen as $\epsilon_{x} = 2$ mm, $\epsilon_{y} = 2$ mm. The upper bound of rotation angle is chosen as $\theta_{\max} = 0.0025$ rad and this gives $\bar{\epsilon}_{x}=1.5$ mm. The selection of these parameters follows directly from the design requirements. The contouring bound $\epsilon_c = 4$~mm corresponds to a $5\%$ tolerance on the circular path radius, set by the application. Since both axes exhibit comparable dynamic characteristics, the bound is equally partitioned as $\epsilon_x = \epsilon_y = 2$~mm. According to Remark~\ref{rem:XTrackBound}, the rotation angle constraint must satisfy $\theta_{\max} \leq \epsilon_x / D$ to ensure a positive tightened bound $\bar{\epsilon}_x > 0$. The value $\theta_{\max} = 0.0025$~rad is selected conservatively to both satisfy this requirement and maintain the system within the valid small-angle linearisation regime, yielding $\bar{\epsilon}_x = 1.5$~mm as the effective X-axis tracking bound.

For the Y-axis control, the switching LTI model \eqref{eq:YnThetaSSModel} is linearised at point $\bar{x}_{h}^{j} = \left\{ -0.075, -0.025, 0.025, 0.075 \right\}$ m for $j \in \mathbb{Z}_{[1,4]}$ with $\Delta x_{h} = 0.025$ m to achieve a trade-off between the accuracy of the control-oriented model and numbers of controller switching.

Two sets of tuning parameters in MPC cost function \eqref{eq:MPC_Y_Theta} and \eqref{eq:MPC_X} are chosen to test the effect of tuning on contouring performance using the proposed control algorithm. The first set of parameters (Tuning A) is chosen with values $Q_{x} = Q_{y}=10^{5}$, $R_{x} = 0.1$ and $R_{y} = \text{diag}(0.1,0.1)$ to put more emphasis on the tracking errors. The second set of tuning parameters (Tuning B) is selected in a less aggressive way with values $Q_{x} = Q_{y} = 10^{3}$, $R_{x} = 0.5$ and $R_{y} = \text{diag}(0.5,0.5)$. The prediction horizon is set to $N = 2$ for both axes MPC. This short horizon is chosen to ensure real-time feasibility at the $2$~ms sampling rate.

It is worth noting that the contouring error guarantee is determined by the RCI set constraints rather than the MPC tuning weights. The contouring error bound $\epsilon_c$ is specified by the manufacturing requirement, and the axis-wise bounds $\epsilon_x$, $\epsilon_y$ are chosen such that $\epsilon_x + \epsilon_y \leq \epsilon_c$, which does not affect the guaranteed bound. The linearisation interval $\Delta x_h$ is a modelling choice. A finer interval provides a closer approximation of the feasible set, making the RCI set computation easier to converge, while the guaranteed performance is maintained as long as the RCI sets exist. The MPC cost weights influence the transient behaviour within the bound, as demonstrated by the two tuning sets, but the bound itself is always satisfied.

For control performance based on Tuning A, the tracking errors on the X, Y-axis and rotation angle are shown in Fig.~\ref{fig:TrackPerfExp1}. The maximum tracking errors of X and Y-axis are $1.35$ mm and $1.23$ mm, respectively, and both values are within the desired tracking error tolerance. The rotation angle is much smaller compared to the upper bound. The corresponding tracking performance based on Tuning B is given in Fig.~\ref{fig:TrackPerfExp2}. The less aggressive tuning results in slightly larger maximum tracking errors on the X and Y-axis as $1.48$ mm and $1.27$ mm respectively, but errors are still within the tolerance. It can be seen that although putting more weight on tracking performance may result in a smaller maximum tracking error, at the same time, it results in more oscillation in tracking.

\begin{figure}[tbp]
	\centerline{\includegraphics[width = 0.85 \columnwidth]{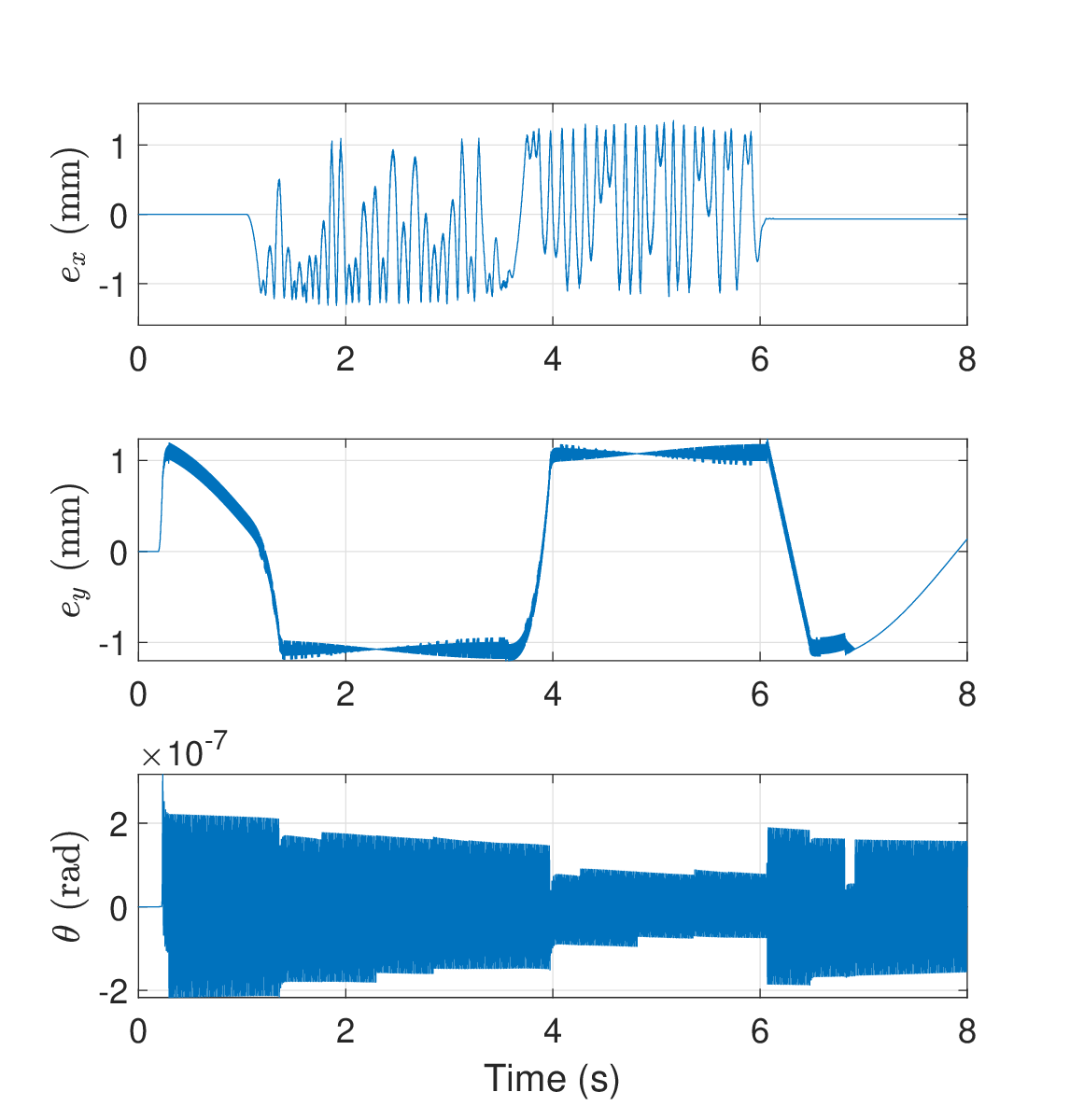}}
	\caption{Experimental tracking performance based on Tuning A: X-axis error (up), Y-axis error (middle) and rotation angle (down).}{\label{fig:TrackPerfExp1}}
\end{figure}

\begin{figure}[tbp]
	\centerline{\includegraphics[width = 0.85 \columnwidth]{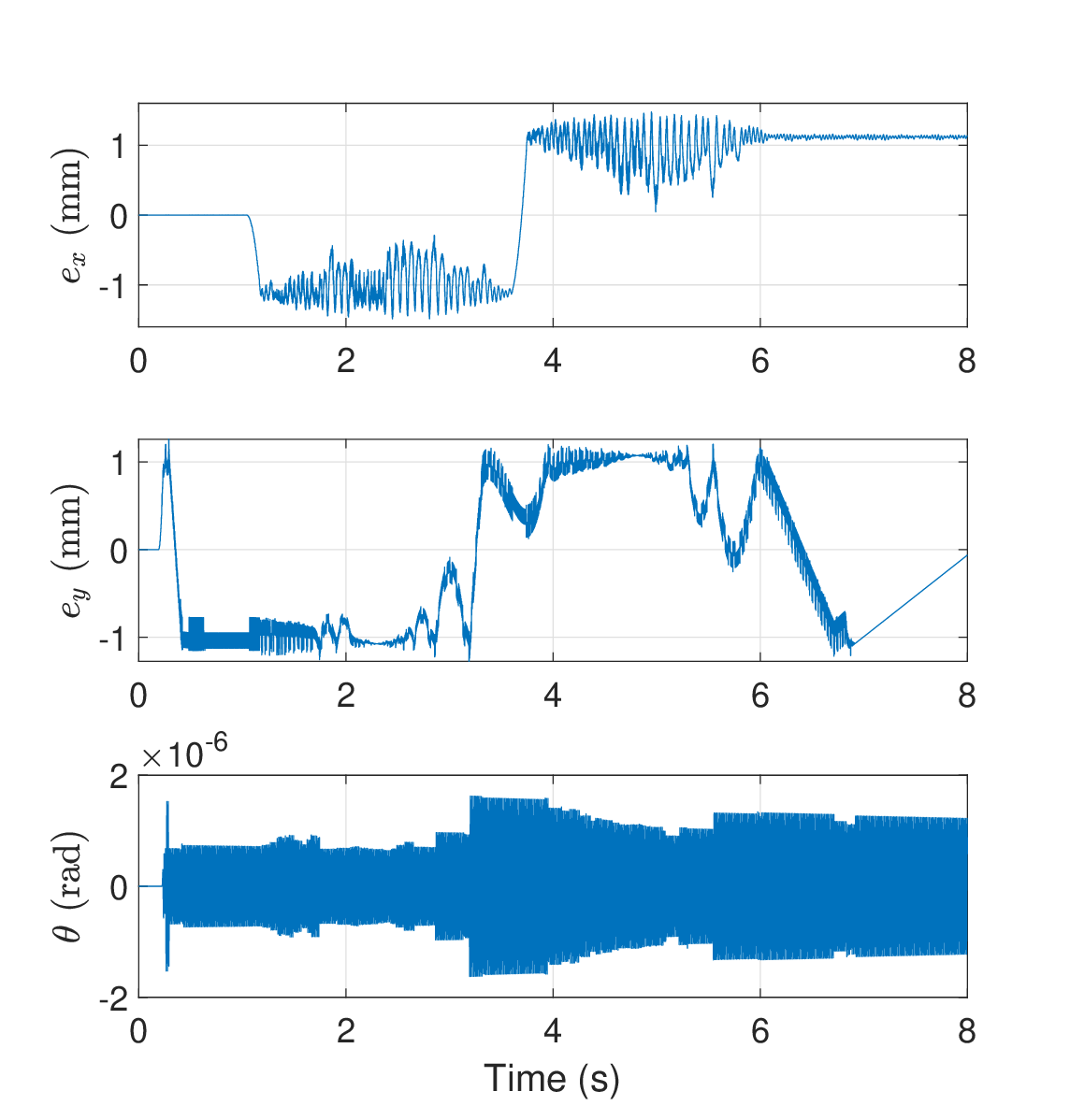}}
	\caption{Experimental tracking performance based on Tuning B: X-axis error (up), Y-axis error (middle) and rotation angle (down).}{\label{fig:TrackPerfExp2}}
\end{figure}

The computed contouring errors based on two tuning parameters are shown in Fig.~\ref{fig:ContErr_Exp}. Since the end-effector position is not directly measured on the real machine, the contouring error reported here is computed from the emulated end-effector states $(x_e, y_e)$, which are reconstructed from the HIL states $(x_h, y_n, \theta)$ using \eqref{eq:xe_position} and \eqref{eq:ye_position}. The maximum contouring errors are $1.65$ mm and $1.62$ mm for Tuning A and B respectively. Both of the contouring errors are within the desired tolerance, and it shows that by detuning the proposed controller, the contouring error can still be bounded. Although a larger steady-state contouring error occurs using the detuned controller, this steady-state error can be removed by including the integrator augmentation in the MPC formulation \cite{stephens2012model}.

\begin{figure}[tbp]
	\centerline{\includegraphics[width = 0.85 \columnwidth]{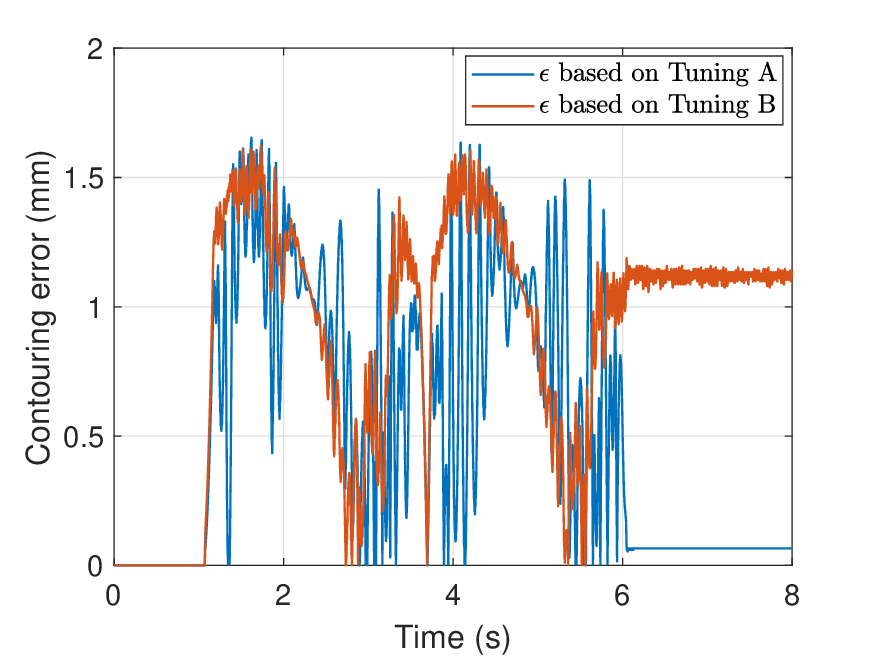}}
	\caption{Experimental contouring error during the whole process.}{\label{fig:ContErr_Exp}}
\end{figure}

\begin{figure}[tbp]
	\centerline{\includegraphics[width = 0.85 \columnwidth]{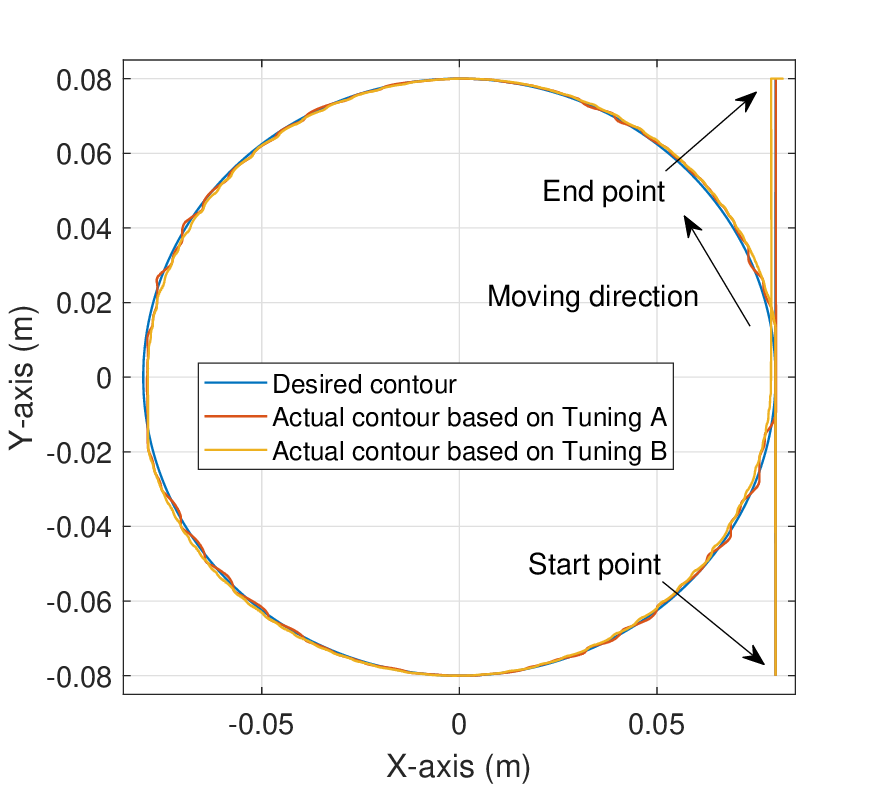}}
	\caption{Desired and actual contour based on different tuning.}{\label{fig:ContPlot_Exp}}
\end{figure}

The comparison of the contouring path achieved by two-different-tuning controllers is shown in Fig.~\ref{fig:ContPlot_Exp}. It is clear that putting more penalties on tracking errors results in more obvious vibration along the path and the proposed error bounded control method can ensure the contour error is not violated for different tuning schemes.

\section{Conclusion}

In this work, a contouring error-bounded control architecture is proposed for dual-drive systems exhibiting position-dependent flexibility and input delay. Experimental results validate the effectiveness of the proposed approach in maintaining prescribed contouring error bounds under input delay and discretisation effects. The presented algorithm shows advantages including its universal applicability for different shapes of desired contour and low commissioning effort. Furthermore, the proposed control framework is extendable to other systems characterised by structural flexibility, enabling contouring error-bounded performance across a wider class of applications. Future work may consider an extension to a linear parameter-varying (LPV) formulation, where the position-dependent dynamics are captured continuously rather than through discrete switching. Combining LPV models with formal contouring error guarantees via RCI sets remains an open theoretical challenge. The generalisation to multi-axis systems is another natural direction for future investigation.

\section*{Acknowledgments}
The authors would like to thank Prof Chris Manzie from University of Melbourne and Prof Iman Shames from Australian National University for their enlightening discussion. We would also like to express our gratitude to Dr Lu Gan at ANCA Motion for his insightful discussion from a practical perspective.

\bibliographystyle{bib/IEEEtran}
\bibliography{bib/TIE_ref}

\begin{thebibliography}{10}
\providecommand{\url}[1]{#1}
\csname url@samestyle\endcsname
\providecommand{\newblock}{\relax}
\providecommand{\bibinfo}[2]{#2}
\providecommand{\BIBentrySTDinterwordspacing}{\spaceskip=0pt\relax}
\providecommand{\BIBentryALTinterwordstretchfactor}{4}
\providecommand{\BIBentryALTinterwordspacing}{\spaceskip=\fontdimen2\font plus
\BIBentryALTinterwordstretchfactor\fontdimen3\font minus
  \fontdimen4\font\relax}
\providecommand{\BIBforeignlanguage}[2]{{%
\expandafter\ifx\csname l@#1\endcsname\relax
\typeout{** WARNING: IEEEtran.bst: No hyphenation pattern has been}%
\typeout{** loaded for the language `#1'. Using the pattern for}%
\typeout{** the default language instead.}%
\else
\language=\csname l@#1\endcsname
\fi
#2}}
\providecommand{\BIBdecl}{\relax}
\BIBdecl

\bibitem{yuan2019bounded}
M.~Yuan, C.~Manzie, M.~Good, I.~Shames, F.~Keynejad, and T.~Robinette,
  ``Bounded error tracking control for contouring systems with end effector
  measurements,'' in \emph{2019 IEEE Internati onal Conference on Industrial
  Technology (ICIT)}.\hskip 1em plus 0.5em minus 0.4em\relax IEEE, 2019, pp.
  66--71.

\bibitem{wang2021global}
W.~Wang, J.~Ma, Z.~Cheng, X.~Li, C.~W. de~Silva, and T.~H. Lee, ``Global
  iterative sliding mode control of an industrial biaxial gantry system for
  contouring motion tasks,'' \emph{IEEE/ASME Transactions on Mechatronics},
  vol.~27, no.~3, pp. 1617--1628, 2021.

\bibitem{li2018modeling}
C.~Li, B.~Yao, and Q.~Wang, ``Modeling and synchronization control of a dual
  drive industrial gantry stage,'' \emph{IEEE/ASME Transactions on
  Mechatronics}, vol.~23, no.~6, pp. 2940--2951, 2018.

\bibitem{sun2024composite}
W.~Sun, J.~Liu, and H.~Gao, ``A composite high-speed and high-precision
  positioning approach for dual-drive gantry stage,'' \emph{IEEE Transactions
  on Automation Science and Engineering}, 2024.

\bibitem{koren1980cross}
Y.~Koren, ``Cross-coupled biaxial computer control for manufacturing systems,''
  \emph{Journal of Dynamic Systems, Measurement, and Control}, vol. 102, no.~4,
  pp. 265--272, Dec. 1980.

\bibitem{aarnoudse2022cross}
L.~Aarnoudse, J.~Kon, K.~Classens, M.~van Meer, M.~Poot, P.~Tacx,
  N.~Strijbosch, and T.~Oomen, ``Cross-coupled iterative learning control for
  complex systems: A monotonically convergent and computationally efficient
  approach,'' in \emph{2022 IEEE 61st Conference on Decision and Control
  (CDC)}.\hskip 1em plus 0.5em minus 0.4em\relax IEEE, 2022, pp. 1485--1490.

\bibitem{koren1991variable}
Y.~Koren and C.-C. Lo, ``Variable-gain cross-coupling controller for
  contouring,'' \emph{CIRP annals}, vol.~40, no.~1, pp. 371--374, 1991.

\bibitem{kuang2019precise}
Z.~Kuang, X.~Li, H.~Wang, H.~Gao, and G.~Sun, ``Precise variable-gain
  cross-coupling contouring control for linear motor direct-drive table,'' in
  \emph{2019 American Control Conference (ACC)}.\hskip 1em plus 0.5em minus
  0.4em\relax IEEE, 2019, pp. 5737--5742.

\bibitem{li2022design}
L.~Li, N.~Cheung, G.~Yang, C.~Wang, P.~Fu, G.~Li, and J.~Pan, ``Design of a
  variable-gain adjacent cross-coupled controller for coordinated motion of
  multiple permanent magnet linear synchronous motors,'' \emph{Computers and
  Electronics in Agriculture}, vol. 192, p. 106561, 2022.

\bibitem{hu2021novel}
N.-T. Hu, L.-Y. Chen, and C.-S. Chen, ``Novel cross-coupling position command
  shaping controller using {H}$\infty$ in multiaxis motion systems,''
  \emph{IEEE Transactions on Industrial Electronics}, vol.~69, no.~12, pp.
  13\,099--13\,110, 2021.

\bibitem{aarnoudse2024cross}
L.~Aarnoudse, J.~Kon, K.~Classens, M.~van Meer, M.~Poot, P.~Tacx,
  N.~Strijbosch, and T.~Oomen, ``Cross-coupled iterative learning control: A
  computationally efficient approach applied to an industrial flatbed
  printer,'' \emph{Mechatronics}, vol.~99, p. 103170, 2024.

\bibitem{zhao2024cross}
J.~Zhao, S.~Gao, Z.~Pan, L.~Wang, and Z.~Yu, ``Cross-coupled synchronous
  control of dual linear motor servo system based on disturbance-assignment
  observer and iterative learning control,'' \emph{IEEE Transactions on
  Transportation Electrification}, 2024.

\bibitem{dao2021high}
H.~V. Dao, S.~Na, D.~G. Nguyen, and K.~K. Ahn, ``High accuracy contouring
  control of an excavator for surface flattening tasks based on extended state
  observer and task coordinate frame approach,'' \emph{Automation in
  Construction}, vol. 130, p. 103845, 2021.

\bibitem{liu2022robust}
Y.~Liu, W.~Sun, C.~Buccella, and C.~Cecati, ``Robust control of
  dual-linear-motor-driven gantry stage for coordinated contouring tasks based
  on feed velocity,'' \emph{IEEE Transactions on Industrial Electronics},
  vol.~70, no.~6, pp. 6229--6238, 2022.

\bibitem{chen2010coordinate}
C.-L. Chen and C.-C. Peng, ``Coordinate transformation based contour following
  control for robotic systems,'' in \emph{Advances in Robot
  Manipulators}.\hskip 1em plus 0.5em minus 0.4em\relax IntechOpen, 2010.

\bibitem{yao2011orthogonal}
B.~Yao, C.~Hu, and Q.~Wang, ``An orthogonal global task coordinate frame for
  contouring control of biaxial systems,'' \emph{IEEE/ASME Transactions on
  Mechatronics}, vol.~17, no.~4, pp. 622--634, 2011.

\bibitem{lou2013task}
Y.~Lou, H.~Meng, J.~Yang, Z.~Li, J.~Gao, and X.~Chen, ``Task polar coordinate
  frame-based contouring control of biaxial systems,'' \emph{IEEE Transactions
  on Industrial Electronics}, vol.~61, no.~7, pp. 3490--3501, 2013.

\bibitem{yang2019novel}
X.~Yang, R.~Seethaler, C.~Zhan, D.~Lu, and W.~Zhao, ``A novel contouring error
  estimation method for contouring control,'' \emph{IEEE/ASME Transactions on
  Mechatronics}, vol.~24, no.~4, pp. 1902--1907, 2019.

\bibitem{lam2012model}
D.~Lam, C.~Manzie, and M.~C. Good, ``Model predictive contouring control for
  biaxial systems,'' \emph{IEEE Transactions on Control Systems Technology},
  vol.~21, no.~2, pp. 552--559, 2012.

\bibitem{yang2015pre}
S.~Yang, A.~H. Ghasemi, X.~Lu, and C.~E. Okwudire, ``Pre-compensation of servo
  contour errors using a model predictive control framework,''
  \emph{International Journal of Machine Tools and Manufacture}, vol.~98, pp.
  50--60, 2015.

\bibitem{li2022control}
Z.~Li, J.~Wang, J.~An, Q.~Zhang, Y.~Zhu, H.~Liu, and H.~Sun, ``Control strategy
  of biaxial variable gain cross-coupled permanent magnet synchronous linear
  motor based on mpc-mras,'' \emph{IEEE Transactions on Industry Applications},
  vol.~58, no.~4, pp. 4733--4743, 2022.

\bibitem{wang2015effect}
L.~Wang, H.~Liu, L.~Yang, J.~Zhang, W.~Zhao, and B.~Lu, ``The effect of axis
  coupling on machine tool dynamics determined by tool deviation,''
  \emph{International Journal of Machine Tools and Manufacture}, vol.~88, pp.
  71--81, 2015.

\bibitem{yuan2019modelling}
M.~Yuan, C.~Manzie, L.~Gan, M.~Good, and I.~Shames, ``Modelling and contouring
  error bounded control of a biaxial industrial gantry machine,'' in \emph{2019
  IEEE Conference on Control Technology and Applications (CCTA)}.\hskip 1em
  plus 0.5em minus 0.4em\relax IEEE, 2019, pp. 388--393.

\bibitem{yuan2026contouring}
M.~Yuan, Y.~Wang, C.~Manzie, Z.~Xu, and T.~Chai, ``Contouring error-bounded
  control for biaxial switched linear systems,'' \emph{IEEE Transactions on
  Systems, Man, and Cybernetics: Systems}, 2026.

\bibitem{yuan2019error}
M.~Yuan, C.~Manzie, M.~Good, I.~Shames, L.~Gan, F.~Keynejad, and T.~Robinette,
  ``Error-bounded reference tracking mpc for machines with structural
  flexibility,'' \emph{IEEE Transactions on Industrial Electronics}, vol.~67,
  no.~10, pp. 8143--8154, 2019.

\bibitem{herceg2013multi}
M.~Herceg, M.~Kvasnica, C.~N. Jones, and M.~Morari, ``Multi-parametric toolbox
  3.0,'' in \emph{2013 European control conference (ECC)}.\hskip 1em plus 0.5em
  minus 0.4em\relax IEEE, 2013, pp. 502--510.

\bibitem{mattingley2012cvxgen}
J.~Mattingley and S.~Boyd, ``Cvxgen: A code generator for embedded convex
  optimization,'' \emph{Optimization and Engineering}, vol.~13, pp. 1--27,
  2012.

\bibitem{stephens2012model}
M.~A. Stephens, C.~Manzie, and M.~C. Good, ``Model predictive control for
  reference tracking on an industrial machine tool servo drive,'' \emph{IEEE
  Transactions on Industrial Informatics}, vol.~9, no.~2, pp. 808--816, 2012.

\end{thebibliography}

\begin{IEEEbiography}[{\includegraphics[width=1in,height=1.25in,clip,keepaspectratio]{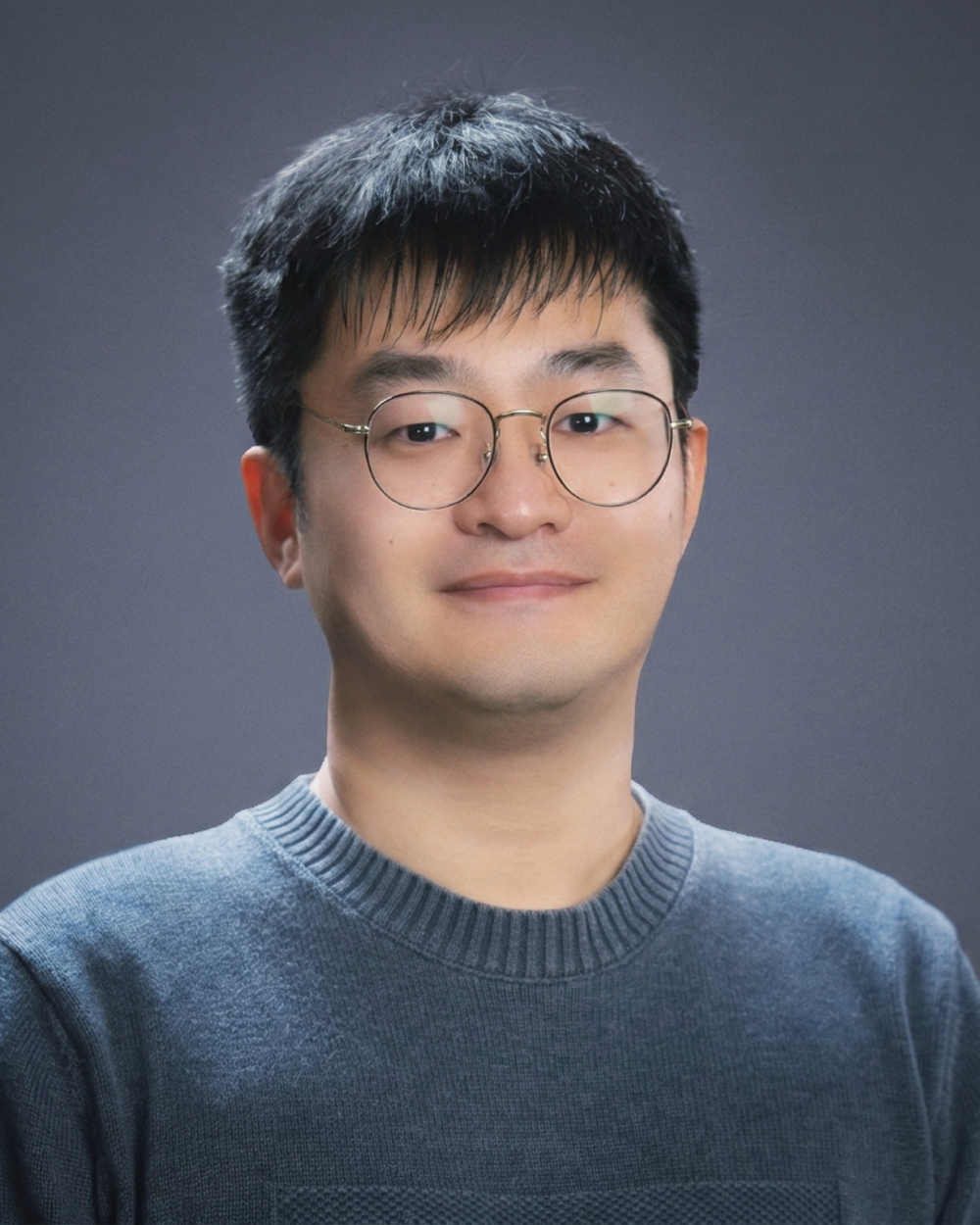}}]{Meng Yuan} (Member, IEEE) received the Ph.D. degree in electrical and electronic engineering from the University of Melbourne, Australia in 2020. He was a Research Fellow at the Rehabilitation Research Institute of Singapore, Nanyang Technological University, Singapore. From 2023 to 2025, he was a Researcher and Marie Skłodowska-Curie Fellow at Chalmers University of Technology, Sweden. He is currently a Lecturer (Assistant Professor) in the School of Engineering and Computer Science at Victoria University of Wellington, New Zealand. His research interests include modeling, learning-based control and model predictive control for batteries, energy systems, mechatronics, and robotics.
\end{IEEEbiography}

\begin{IEEEbiography}[{\includegraphics[width=1in,height=1.25in,clip,keepaspectratio]{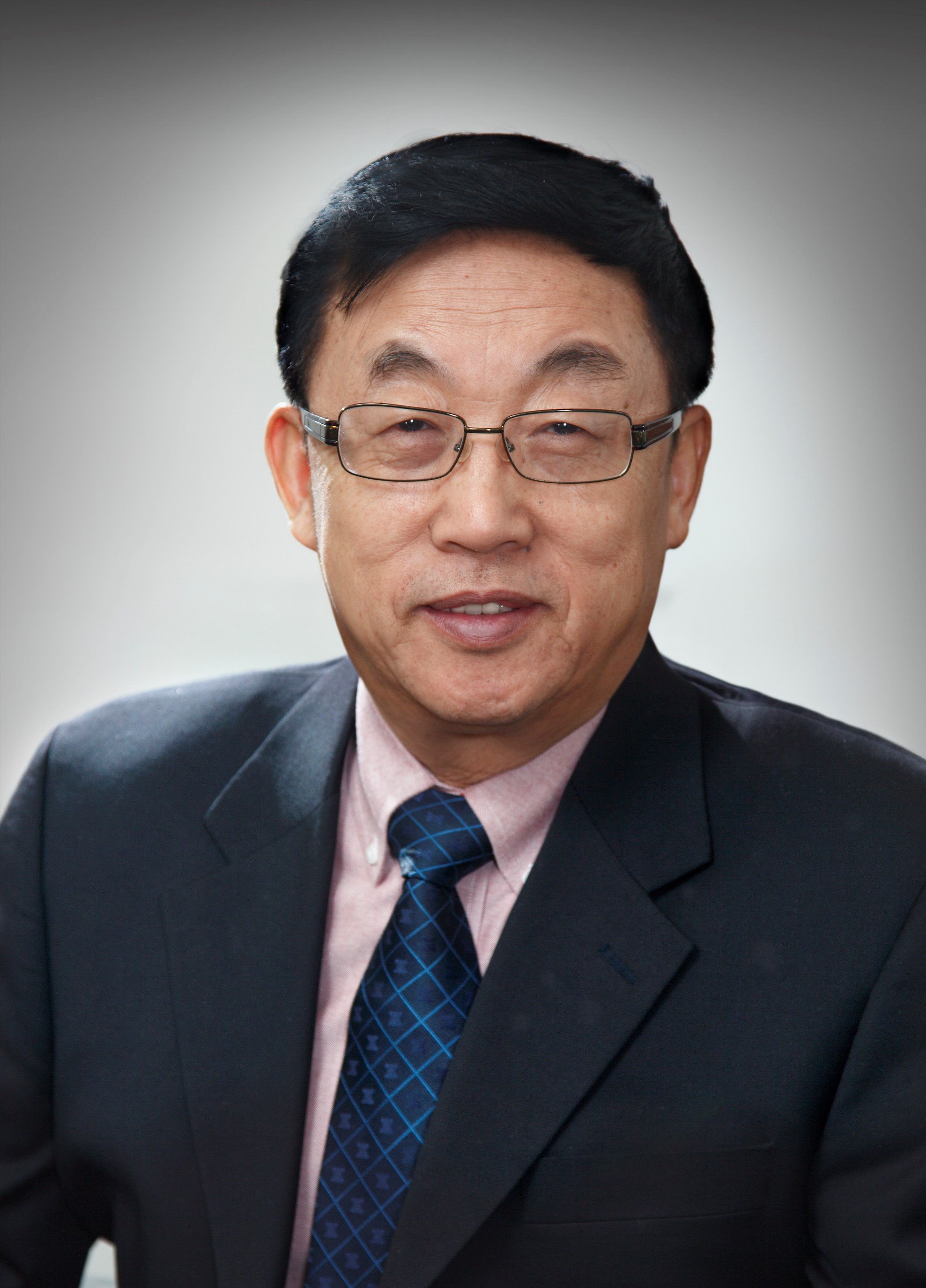}}]{Tianyou Chai} (Life Fellow, IEEE) received the Ph.D. degree in control theory and engineering in 1985 from Northeastern University, Shenyang, China, where he became a Professor in 1988. He is the founder and Director of the Center of Automation, which became a National Engineering and Technology Research Center and a State Key Laboratory. He is a member of Chinese Academy of Engineering, IFAC Fellow and IEEE Fellow. He has served as director of Department of Information Science of National Natural Science Foundation of China from 2010 to 2018. His current research interests include modeling, control, optimization and integrated automation and intelligence of complex industrial processes.
\end{IEEEbiography}

\vfill

\end{document}